\ifx\SubmitVer\undefined
   \newcommand{\InSubmitVer}[1]{}%
   \newcommand{\InNotSubmitVer}[1]{#1}%
\else
   
   \newcommand{\InSubmitVer}[1]{#1}%
   \newcommand{\InNotSubmitVer}[1]{}%
\fi

\InSubmitVer{%
   \documentclass[11pt]{article}%
}
\InNotSubmitVer{%
   \documentclass[12pt]{article}%
}

\usepackage{amsmath}
\usepackage{amssymb}
\usepackage[cm]{fullpage}

\usepackage{paralist}
\usepackage{graphicx}   
\usepackage{caption}
\usepackage{color}
\usepackage{xspace}
\usepackage{picins}

\usepackage[cmyk]{xcolor}

\usepackage{hyperref}%
\hypersetup{%
   breaklinks,%
   ocgcolorlinks,
   colorlinks=true,%
   linkcolor=[rgb]{0.45,0.0,0.0},%
   citecolor=[rgb]{0,0,0.45}
}

\usepackage{euscript}%
\usepackage{amsmath}%
\usepackage[amsmath,amsthm,thmmarks]{ntheorem}
\theoremseparator{.}%

\usepackage{titlesec}
\titlelabel{\thetitle. }

\InSubmitVer{%
   \usepackage{lineno}%
}

\numberwithin{figure}{section}
\numberwithin{table}{section}
\numberwithin{equation}{section}

\newtheorem{theorem}{Theorem}[section] 
\newtheorem{lemma}[theorem]{Lemma}

{\theorembodyfont{\rm}%
   \newtheorem{defn}[theorem]{Definition}%
   \newtheorem{observation}[theorem]{Observation}
   
}

\newlength{\savedparindent}

\renewcommand{\Re}{{\rm I\!\hspace{-0.025em} R}}
\newcommand{\figlab}[1]{\label{fig:#1}}
\newcommand{\figref}[1]{Figure~\ref{fig:#1}}

\newcommand{\eqlab}[1]{\label{equation:#1}}
\newcommand{\Eqref}[1]{Eq.~(\ref{equation:#1})}
\newcommand{\Eqrefpage}[1]{Eq.~(\ref{equation:#1})%
   $_\text{p\pageref{equation:#1}}$}

\newcommand{\apndlab}[1]{\label{apnd:#1}}
\newcommand{\apndref}[1]{Appendix~\ref{apnd:#1}}

\newcommand{\lemlab}[1]{\label{lemma:#1}}
\newcommand{\lemref}[1]{Lemma~\ref{lemma:#1}}

\newcommand{\obslab}[1]{\label{observation:#1}}
\newcommand{\obsref}[1]{Observation~\ref{observation:#1}}

\newcommand{\seclab}[1]{\label{section:#1}}
\newcommand{\secref}[1]{Section~\ref{section:#1}}
\newcommand{\thmlab}[1]{\label{theorem:#1}}
\newcommand{\thmref}[1]{Theorem~\ref{theorem:#1}}
\newcommand{\thmrefpage}[1]{Theorem~\ref{theorem:#1}%
   $_\text{p\pageref{theorem:#1}}$}

\providecommand{\deflab}[1]{\label{def:#1}}
\newcommand{\defref}[1]{Definition~\ref{def:#1}}

\newcommand{\blockerY}[2]{\mathrm{blocker{}Set}\pth{#1, #2}}

\newcommand{\CapSet}{\mathcal{C}}
\newcommand{\coneY}[2]{\mathrm{cone}\pth{#1, #2}}

\newcommand{\Term}[1]{\textsf{#1}}
\newcommand{\PTAS}{\Term{PTAS}\xspace}

\DefineNamedColor{named}{RedViolet} {cmyk}{0.07,0.90,0,0.34}

\newcommand{\AlgorithmI}[1]{{\textcolor[named]{RedViolet}{\texttt{\bf{#1}}}}}
\newcommand{\Algorithm}[1]{{\AlgorithmI{#1}\index{algorithm!#1@{\AlgorithmI{#1}}}}}

\newcommand{\GreedySeparator}{\Algorithm{GreedySeparator}\xspace}

\newcommand{\etal}{\textit{et~al.}\xspace}

\newcommand{\inducedBall}[1]{\mathrm{ball}_{\mathrm{in}}\pth{#1}}
\newcommand{\Ring}{\mathcal{R}}

\definecolor{blue25}{rgb}{0,0,0.85}%
\newcommand{\emphic}[2]{%
   \textcolor{blue25}{%
      \textbf{\emph{#1}}}%
   \index{#2}}

\newcommand{\emphi}[1]{\emphic{#1}{#1}}

\newcommand{\pth}[2][\!]{#1\left({#2}\right)}%

\newcommand{\VorX}[1]{\mathcal{V} \pth{#1}}%
\newcommand{\VorCell}[2]{\mathcal{C}_{#2} \pth{#1}}%
\newcommand{\brc}[1]{\left\{ {#1} \right\}}%
\newcommand{\sep}[1]{\,\left|\, {#1} \MakeBig\right.}%
\newcommand{\dist}[2]{\left\| {#1} - {#2}  \right\|}%

\newcommand{\distChar}{\mathsf{d}}%
\newcommand{\distSet}[2]{\distChar\pth{#1, #2}}
\newcommand{\MakeBig} {\rule[-.2cm]{0cm}{0.4cm}}
\newcommand{\MakeSBig}{\rule[0.0cm]{0.0cm}{0.39cm}} 
\newcommand{\MakesBig}{\rule[0.0cm]{0.0cm}{0.385cm}} 

\newcommand{\bisectorX}[2]{\beta_{#1, #2}}
\newcommand{\pencilX}[2]{\mathrm{pencil}\pth{#1, #2}}
\newcommand{\pencilF}[1]{\mathrm{pencil}\pth{#1}}
\newcommand{\pencilPnt}[2]{\mathrm{pencil}_{#1}\pth{#2}}
\newcommand{\pencilSet}[2]{\mathrm{pencil}_{#1}\pth{#2}}

\newcommand{\trailX}[1]{\mathrm{trail}\pth{#1}}

\newcommand{\mbX}[1]{\mathrm{mb}\pth{#1}}

\newcommand{\SarielThanks}[1]{%
   \thanks{%
      Department of Computer Science; %
      University of Illinois; %
      201 N. Goodwin Avenue; %
      Urbana, IL, 61801, USA; %
      {\tt \url{http://sarielhp.org}.} %
      #1%
   }
}

\newcommand{\VijayThanks}[1]{%
   \thanks{%
      Department of Computer Science; %
      University of Illinois; %
      201 N. Goodwin Avenue; %
      Urbana, IL, 61801, USA; %
      {\tt \url{http://vspvijay.com}.} %
      #1
   }
}

\newcommand{\Ex}[2][\!]{\mathop{\mathbf{E}}#1\pbrcx{#2}}
\newcommand{\Prob}[1]{\mathop{\mathbf{Pr}}\!\pbrcx{#1}}

\newcommand{\cardin}[1]{\left| {#1} \right|}%
\newcommand{\si}[1]{#1}

\newcommand{\CHX}[1]{\mathcal{CH}\pth{#1}}
\newcommand{\eps}{{\varepsilon}}%

\newcommand{\PntSet}{\mathsf{P}}%
\newcommand{\PntSetB}{\mathsf{Q}}%

\newcommand{\SetL}{\mathsf{L}}%
\newcommand{\SetO}{\mathsf{O}}%

\newcommand{\SetX}{\mathsf{X}}%
\newcommand{\SetY}{\mathsf{Y}}%

\newcommand{\SepSet}{\mathcal{Z}}%

\newcommand{\pnt}{\mathsf{p}}%
\newcommand{\pntA}{\mathsf{q}}%
\newcommand{\pntB}{\mathsf{s}}%
\newcommand{\pntC}{\mathsf{t}}%

\newcommand{\setA}{\mathsf{X}}
\newcommand{\setB}{\mathsf{Y}}

\newcommand{\ballX}[2]{\mathrm{ball}\pth{#1, #2}}

\newcommand{\ballA}{b}
\newcommand{\ballB}{b'}

\newcommand{\leftover}{f}
\newcommand{\Leftover}{F}

\newcommand{\cenA}{\psi}

\newcommand{\sphereB}{\ensuremath{\mathbb{S}}}

\newcommand{\radA}{{r}}
\newcommand{\radB}{{r}'}

\providecommand{\pbrc}[2][\!\!]{#1\left[ {#2} \MakeBig \right]}
\providecommand{\pbrcS}[2][]{\left[ {#2} \right]}
\providecommand{\pbrcx}[1]{\left[ {#1} \right]}

\providecommand{\ceil}[1]{\left\lceil {#1} \right\rceil}

\newcommand{\radN}{\ell}

\newcommand{\dblCd}{c^d_{\mathrm{dbl}}}
\newcommand{\constSep}{c_{\mathrm{sep}}}

\newcommand{\nGuard}{\mathrm{g}}
\newcommand{\maxGuard}{\mathsf{b}}

\newcommand{\HH}{\mathcal{H}}%
\newcommand{\HHB}{\mathcal{G}}%

\newcommand{\PP}{\mathcal{P}}

\newcommand{\grd}{\mathsf{g}}

\newcommand{\BadBallsA}{\mathcal{B}}
\newcommand{\BadBalls}[3]{\mathcal{B}\pth{#1, #2, #3}}
\newcommand{\violatorX}[3]{\mathcal{BP}\pth{#1, #2, #3}}
\newcommand{\bFeaturesX}[4][\!]{\mathcal{F}_{\mathrm{bad}}\pth[#1]{#2, #3,
      #4}}

\newcommand{\Grid}{\mathsf{G}}
\newcommand{\Cell}{C}
\newcommand{\projY}[2]{\mathsf{nn}\pth{#1,#2}}

\newcommand{\lengthX}[1]{\mathrm{len}\pth{#1}}

\newcommand{\feature}{\mathsf{f}}
\newcommand{\sitesX}[1]{\mathrm{sites}\pth{#1}}

\newcommand{\spanF}[1]{\mathrm{span}\pth{#1}}

\newcommand{\halfspaceX}[2]{\mathcal{H}\pth{#1,#2}}

\newcommand{\shellF}[1]{\mathrm{shell}\pth{#1}}

\newcommand{\indet}[2]{\mathcal{ID}\pth{#1,#2}}

\newcommand{\annihilateF}[2]{\mathcal{A}_{#1}\pth{#2}}

\newcommand{\Flat}{\mathsf{F}}

\newcommand{\annihilateX}[1]{\mathcal{A}\pth{#1}}

\newcommand{\Graph}{\mathsf{G}}%
\newcommand{\Vertices}{\mathsf{V}}%
\newcommand{\Edges}{\mathsf{E}}

\newcommand{\hitSetSizeX}[1]{\nu\pth{#1}}

\renewcommand{\th}{\si{th}\xspace}

\newcommand{\nnX}[2]{\mathsf{nn}\pth{#1, #2}}

\DefineNamedColor{named}{OliveGreen}{cmyk}{0.64, 0, 0.95, 0.40}

\providecommand{\ComplexityClass}[1]{{{\textcolor[named]{OliveGreen}{%
      \textsc{#1}}}}}

\providecommand{\NPHard}{{\ComplexityClass{NP-Hard}}%
   \index{NP!hard}\xspace}

\newcommand{\remove}[1]{}

\newcommand{\oSepSetB}{T}%

\newcommand{\oSepSet}{\SepSet^{\sphereB}}%
\newcommand{\SetXX}{\Xi}

\newcommand{\pitemlab}[1]{\label{p:item:#1}}
\newcommand{\pitemref}[1]{(P\ref{p:item:#1})}

\newcommand{\pitemrefpage}[1]{(P\ref{p:item:#1})%
   $_\text{p\pageref{p:item:#1}}$}

\newcommand{\lSize}{\ell}

\newcommand{\Lopt}{\mathcal{L}}
\newcommand{\Opt}{\mathcal{O}}

\newcommand{\SetBalls}{\EuScript{B}}%

\newcommand{\ArrX}[1]{\mathcal{A}\pth{#1}}%

\newcommand{\Assign}{:=}

\providecommand{\CodeComment}[1]{\textcolor{blue}{\texttt{#1}}}

\newenvironment{myprogram}{
   \begin{minipage}{4.0in}
   \begin{tabbing}
       \ \ \ \ \= \ \ \ \= \ \ \ \ \= \ \ \ \ \= \ \ \ \ \=
      \ \ \ \ \= \ \ \ \ \= \ \ \ \ \= \ \ \ \ \=
      \ \ \ \ \= \ \ \ \ \= \ \ \ \ \= \ \ \ \ \= \kill
}{
   \end{tabbing}
   \end{minipage}
}

\newcommand{\AlgorithmBox}[1]{
   \fbox{\begin{myprogram} 
          #1
          \end{myprogram}
       }
   }
\newcommand{\Do}{{\small\bf do}\ }

\newcommand{\While}{{\small\bf while}\ }

\newcommand{\algClusterLocal}{\Algorithm{ClusterLocalOpt}\xspace}

\newcommand{\lChuckSize}{\lambda}

\newcommand{\Iterations}{I}

\begin{document}

\InSubmitVer{\linenumbers}

\title{Separating a Voronoi Diagram via Local Search%
   \footnote{%
      Work on this paper by the second author was partially supported
      by NSF AF award CCF-0915984, and NSF AF award CCF-1217462.}%
}%

\author{%
   Vijay V. S. P.  Bhattiprolu%
   \VijayThanks{}%
   \and%
   Sariel Har-Peled%
   \SarielThanks{}%
}%
\date{\today}

\maketitle

\begin{abstract}
    Given a set $\PntSet$ of $n$ points in $\Re^d$, we show how to
    insert a set $\SetX$ of $O\pth{ n^{1-1/d} }$ additional points,
    such that $\PntSet$ can be broken into two sets $\PntSet_1$ and
    $\PntSet_2$, of roughly equal size, such that in the Voronoi
    diagram $\VorX{\PntSet \cup \SetX}$, the cells of $\PntSet_1$ do
    not touch the cells of $\PntSet_2$; that is, $\SetX$ separates
    $\PntSet_1$ from $\PntSet_2$ in the Voronoi diagram (or in the
    dual Delaunay triangulation).  Given such a partition
    $(\PntSet_1,\PntSet_2)$ of $\PntSet$, we present approximation
    algorithms to compute the minimum size separator realizing this
    partition.

    Finally, we present a simple local search algorithm that is a
    \PTAS for geometric hitting set of fat objects (which can also be
    used to approximate the optimal Voronoi partition).
\end{abstract}

\InSubmitVer{%
   \thispagestyle{empty}%
   \newpage%
   \setcounter{page}{1}%
}

\section{Introduction}

\paragraph{Divide and conquer.}
Many algorithms work by partitioning the input into a small number of
pieces, of roughly equal size, with little interaction between the
different pieces, and then recursing on these pieces. One natural way
to compute such partitions for graphs is via the usage of separators.

\paragraph{Separators.}
A (vertex) separator of a graph $\Graph = (\Vertices,\Edges)$,
informally, is a ``small'' set $\SepSet \subseteq \Vertices$ whose
removal breaks the graph into two or more connected subgraphs, each of
which is of size at most $n/c$, where $c$ is some constant larger than
one. As a concrete example, any tree with $n$ vertices has a single
vertex, which can be computed in linear time, such that its removal
breaks the tree into subtrees, each with at most $n/2$ vertices.

\paragraph{Separators in planar graphs.}

In 1977, Lipton and Tarjan \cite{lt-stpg-77, lt-stpg-79} proved that
any planar graph with $n$ vertices contains a separator of size
$O\pth{\sqrt{n}}$, and it can be computed in linear time.
Specifically, there exists a separator of size $O(\sqrt{n})$ that
partitions the graph into two disjoint subgraphs each containing at
most $2n/3$ vertices.

There has been a substantial amount of work on planar separators in
the last four decades, and they are widely used in data-structures and
algorithms for planar graphs, including
\begin{inparaenum}[(i)]
    \item shortest paths \cite{fr-pgnwe-06},
    \item distance oracles \cite{svy-dosg-09},
    \item max flow \cite{ek-ltamf-13}, and
    \item approximation algorithms for TSP \cite{k-ltast-08}.
\end{inparaenum}
This list is a far cry from being exhaustive, and is a somewhat
arbitrary selection of some recent results on the topic.

\paragraph{Planar separators via geometry.}

Any planar graph can be realized as a set of interior disjoint disks,
where a pair of disks touch each other, if and only if the
corresponding vertices have an edge between them. This is known as the
circle packing theorem \cite{pa-cg-95}, sometimes referred to in the
literature as Koebe-Andreev-Thurston theorem. Its original proof goes
back to Koebe's work in 1936 \cite{k-kdka-36}.

Surprisingly, the existence of a planar separator is an easy
consequence of the circle packing theorem. This was proved by Miller
\etal \cite{mttv-sspnng-97}, and their proof was recently simplified
by Har-Peled \cite{h-speps-11}. Among other things, Miller \etal
showed that given a set of $n$ balls in $\Re^d$, such that no point is
covered more than $k$ times, the intersection graph of the balls has a
separator of size $O\pth{k^{1/d} n^{1-1/d}}$. This in turn implies
that the $k$-nearest neighbor graph of a set of points in $\Re^d$, has
a small separator \cite{mttv-sspnng-97, h-speps-11}.  Various
extensions of this technique were described by Smith and Wormald
\cite{sw-gsta-98}.

\paragraph{Other separators.}

Small separators are known to exist for many other families of
graphs. These include graphs
\begin{inparaenum}[(i)]
    \item with bounded tree width \cite{bptw-betsu-10},
    \item with bounded genus \cite{ght-stgbg-84},
    \item that are minor free \cite{ast-stgem-90}, and
    \item that are grids.
\end{inparaenum}

\begin{figure}
    \centerline{%
       \begin{tabular}{\si{ccccc}}
           \begin{minipage}{0.5cm}
               \vspace{-6cm} (A) \vspace{3cm}
           \end{minipage}
           & \includegraphics[page=1,scale=0.8]{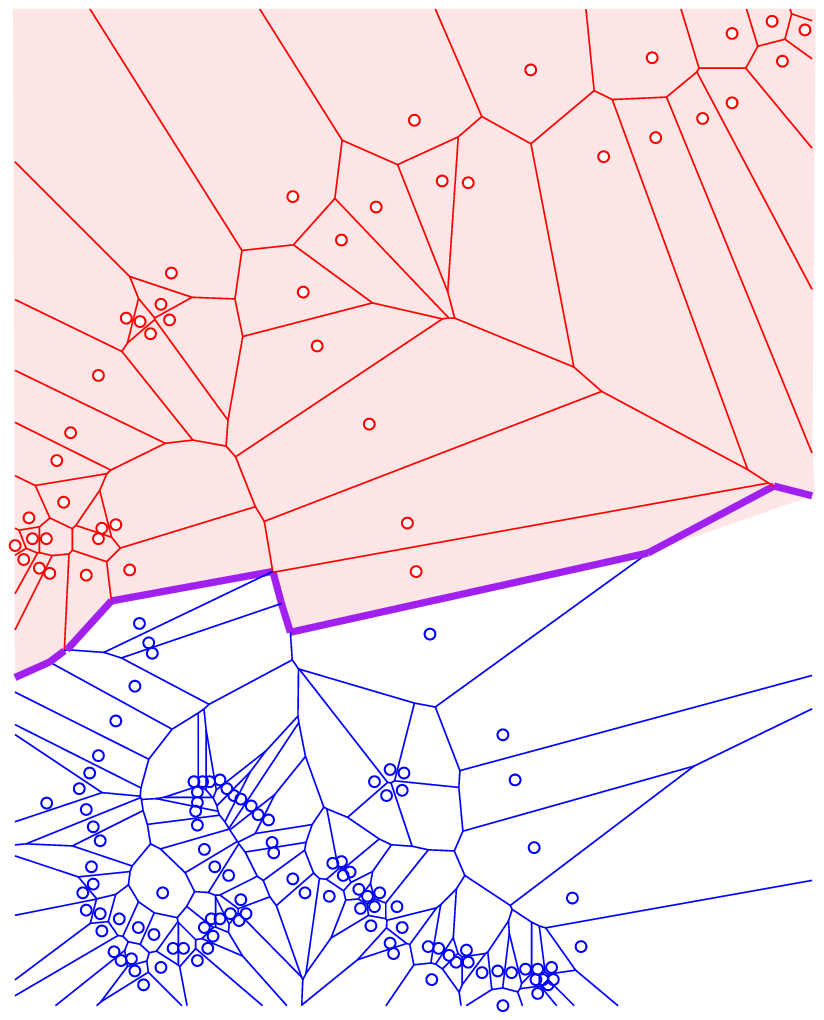}
           &
           &
           \begin{minipage}{0.5cm}
               \vspace{-6cm} (B) \vspace{3cm}
           \end{minipage}%
           &\includegraphics[page=2,scale=0.8]{figs/alligator}
       \end{tabular}%
    }%
    \vspace{-0.5cm}%
    \caption{(A) A Voronoi partition. (B) A separator realizing it.}%
    \figlab{v:partition}
\end{figure}

\paragraph{Voronoi separators.}

In this paper, we are interested in geometric separation in a Voronoi
diagram. Voronoi diagrams are fundamental geometric structure, see
\cite{akl-vddt-13}. Specifically, given a set $\PntSet$ of points in
$\Re^d$, we are interested in inserting a small set of new points
$\SetX$, such that there is a balanced partition of $\PntSet$ into two
sets $\PntSet_1, \PntSet_2$, such that no cell of $\PntSet_1$ touches
a cell of $\PntSet_2$ in the Voronoi diagram $\VorX{\PntSet \cup
   \SetX}$. Note, that such a set $\SetX$ also separates $\PntSet_1$
and $\PntSet_2$ in the Delaunay triangulation of $\PntSet \cup \SetX$.

\paragraph{Why Voronoi separators are interesting?}

Some meshing algorithms rely on computing a Delaunay triangulation of
geometric models to get good triangulations that describe solid
bodies. Such meshes in turn are fed into numerical solvers to simulate
various physical processes. To get good triangulations, one performs a
Delaunay refinement that involves inserting new points into the
triangulations, to guarantee that the resulting elements are well
behaved. Since the underlying geometric models can be quite
complicated and these refinement processes can be computationally
intensive, it is natural to try and break up the data in a balanced
way, and Voronoi separators provide one way to do so.

More generally, small Voronoi separators provide a way to break up a
point set in such a way that there is limited interaction between two
pieces of the data.

\paragraph{Geometric hitting set.}

Given a set of objects in $\Re^d$, the problem of finding a small
number of points that stab all the objects is an instance of geometric
hitting set. There is quite a bit of research on this problem. In
particular, the problem is \NPHard for almost any natural instance,
but a polynomial time $(1+\eps)$-approximation algorithm is known for
the case of balls in $\Re^d$ \cite{c-ptasp-03}, where one is allowed
to place the stabbing points anywhere. The discrete variant of this
problem, where there is a set of allowable locations to place the
stabbing points, seems to be significantly harder and only weaker
results are known \cite{hl-wgscp-12}.

One of the more interesting versions of the geometric hitting set
problem, is the art gallery problem, where one is given a simple
polygon in the plane, and one has to select a set of points (inside or
on the boundary of the polygon) that ``see'' the whole polygon.  While
much research has gone into variants of this problem \cite{o-agta-87},
nothing is known as far as an approximation algorithm (for the general
problem). The difficulty arises from the underlying set system being
infinite, see \cite{eh-ggt-06} for some efforts in better
understanding this problem.

\paragraph{Geometric local search.}

Relatively little is known regarding local search methods for
geometric approximation problems. Arya \etal \cite{agkmp-lshkm-01}
gave a local search method for approximating $k$-median clustering by
a constant factor, and this was recently simplified by Gupta and
Tangwongsan \cite{gt-salsa-08}.

Closer to our settings, Mustafa and Ray \cite{mr-pghsp-09} gave a
local search algorithm for the discrete hitting set problem over
pseudo disks and $r$-admissible regions in the plane, which yields a 
\PTAS.  Chan and Har-Peled \cite{ch-aamis-12} gave a local search \PTAS 
for the independent set problem over fat objects, and for pseudodisks 
in the plane. Both works use separators in proving the quality of 
approximation.


\subsection{Our Results}

In this paper we give algorithms for the following:

\smallskip

\begin{enumerate}[(A)]
    \item \textbf{Computing a small Voronoi separator.} %
    Given a set $\PntSet$ of $n$ points in $\Re^d$, we show how to
    compute, in expected linear time, a balanced Voronoi separator of
    size $O\pth{n^{1-1/d}}$. This is described in
    \secref{balanced:separator}. The existence of such a separator was
    not known before, and our proof is relatively simple and elegant.
    
    Such a separator can be used to break a large data-set into a
    small number of chunks, such that each chunk can be handled
    independently -- potentially in parallel on different computers.

    \item \textbf{Exact algorithm for computing the smallest Voronoi
       separator realizing a given partition.} %
    In \secref{exact:separator}, given a partition $(\PntSet_1,
    \PntSet_2)$ of a point set $\PntSet$ in $\Re^d$, we describe an
    algorithm that computes the minimum size Voronoi separator
    realizing this separation.  The running time of the algorithm is
    $n^{O\pth[]{\maxGuard}}$, where $\maxGuard$ is the cardinality of
    the optimal separating sets (the $O$ notation is hiding a constant
    that depends on $d$).

    \item \textbf{Constant approximation algorithm for the smallest
       Voronoi separator realizing a given partition.} %
    In \secref{const:factor}, we describe how to compute a constant
    factor approximation to the size of the minimal Voronoi separator
    for a given partition of a set in $\Re^d$. This is the natural
    extension of the greedy algorithm for geometric hitting set of
    balls, except that in this case, the set of balls is infinite and
    is encoded implicitly, which somewhat complicates things.

    \item \textbf{A \PTAS for the smallest Voronoi separator realizing
       a given partition.} %
    In \secref{PTAS}, we present a polynomial time approximation
    scheme to compute a Voronoi separator, realizing a given
    partition, whose size is a $(1+\eps)$-approximation to the size of
    the minimal Voronoi separator for a given partition of a set in
    $\Re^d$. The running time is $n^{O(1/\eps^d)}$.
    
    Interestingly, the new algorithm provides a \PTAS for the
    geometric hitting set problem (for balls), that unlike previous
    approaches that worked top-down \cite{c-ptasp-03, ejs-ptasg-05},
    works more in a bottom-up approach. Note, that since our set of
    balls that needs to be pierced is infinite, and is defined
    implicitly, it is not obvious a priori how to use the previous
    algorithms in this case. 

    \medskip%
    \textbf{Sketch of algorithm.} %
    The new algorithm works by first computing a ``dirty'' constant
    approximation hitting set using a greedy approach (this is
    relatively standard). Somewhat oversimplifying, the algorithm next
    clusters this large hitting set into tight clusters of size $k =
    O(1/\eps^d)$ each. It then replaces each such cluster of the weak
    hitting set, by the optimal hitting set that can pierce the same
    set of balls, computed by using the exact algorithm -- which is
    ``fast'' since the number of piercing points is at most
    $O(1/\eps^d)$. In the end of this process the resulting set of
    points is the desired hitting set.  Namely, the new approximation
    algorithm reduces the given geometric hitting set instance, into
    $O(m/k)$ smaller instances where $m$ is the size of the overall
    optimal hitting set and each of the smaller instances has an
    optimal hitting set of size $O(k)$.
    

    For the analysis of this algorithm, we need a strengthened version
    of the separator theorem. See \thmrefpage{ball:separator} for
    details.
    
    \item \textbf{Local search \PTAS for continuous geometric hitting
       set problems.} %
    An interesting consequence of the new bottom-up \PTAS, is that it
    leads to a simple local search algorithm for geometric hitting set
    problems for fat objects. Specifically, in \secref{local:search},
    we show that the algorithm starts with any hitting set (of the
    given objects) and continues to make local improvements via
    exchanges of size at most $O\pth{1/\eps^d}$, until no such
    improvement is possible, yielding a \PTAS.  The analysis of the
    local search algorithm is subtle requiring to cluster
    simultaneously the locally optimal solution, and the optimal
    solution, and matching these clusters to each other.

\end{enumerate}

\paragraph{Significance of Results.}

Our separator result provides a new way to perform geometric divide
and conquer for Voronoi diagrams (or Delaunay triangulations).  The
\PTAS for the Voronoi partition problem makes progress on a geometric
hitting set problem where the ranges to be hit are defined implicitly,
and their number is infinite, thus pushing further the envelope of
what geometric hitting set problems can be solved efficiently. Our
local search algorithm is to our knowledge the first local search
algorithm for geometric hitting set -- it is simple, easy to
implement, and might perform well in practice (this remains to be
verified experimentally, naturally). More importantly, it shows that
local search algorithms are potentially more widely applicable in
geometric settings.

\paragraph{How our results relate to known results?}

Our separator result is similar in spirit (but not in details!) to the
work of Miller \etal \cite{mttv-sspnng-97} on a separator for a
$k$-ply set of balls -- the main difference being that Voronoi cells
behave very differently than balls do. Arguably, our proof is
significantly simpler and more elegant. Our bottom-up \PTAS approach
seems to be new, and should be applicable to other problems. Having
said that, it seems like the top-down approaches \cite{c-ptasp-03,
   ejs-ptasg-05} potentially can be modified to work in this case, but
the low level details seem to be significantly more complicated, and
the difficulty in making them work was the main motivation for
developing the new approach. The analysis of our local search
algorithm seems to be new -- in particular, the idea of incrementally
clustering in sync optimal and local solutions.  Of course, the basic
idea of using separators in analyzing local search algorithms appear
in the work of Mustafa and Ray \cite{mr-pghsp-09} and Chan and
Har-Peled \cite{ch-aamis-12}.




\section{Preliminaries}

For a point set $\PntSet \subseteq \Re^d$, the \emphi{Voronoi diagram}
of $\PntSet$, denoted by $\VorX{\PntSet}$ is the partition of space
into convex cells, where the \emphi{Voronoi cell} of $\pnt \in
\PntSet$ is
\begin{equation*}
    \VorCell{\pnt}{\PntSet} = \brc{\pntA \in \Re^d \sep{
          \dist{\pntA}{\pnt} \leq \distSet{\pntA}{\PntSet}}},
\end{equation*}
where $\distSet{\pntA}{\PntSet} = \min_{\pntB \in \PntSet}
\dist{\pntA}{\pntB}$ is the distance of $\pntA$ to the set $\PntSet$.
Voronoi diagrams are a staple topic in Computational Geometry, see
\cite{bcko-cgaa-08}, and we include the definitions here for the sake
of completeness. In the plane, the Voronoi diagram has linear
descriptive complexity. For a point set $\PntSet$, and points $\pnt,
\pntA \in \PntSet$, the geometric loci of all points in $\Re^d$ that
have both $\pnt$ and $\pntA$ as nearest neighbor, is the
\emphi{bisector} of $\pnt$ and $\pntA$ -- it is denoted by
$\bisectorX{\pnt}{\pntA} = \brc{\pntB \in \Re^d
   \sep{\dist{\pntB}{\pnt} = \dist{\pntB}{\pntA} =
      \distSet{\pntB}{\PntSet}}}$.  A point $\pntB \in
\bisectorX{\pnt}{\pntA}$ is the center of a ball whose interior does
not contain any point of $\PntSet$ and that has $\pnt$ and $\pntA$ on
its boundary. The set of all such balls induced by
$\bisectorX{\pnt}{\pntA}$ is the \emphi{pencil} of $\pnt$ and $\pntA$,
denoted by $\pencilX{\pnt}{\pntA}$.

\begin{defn}
    Let $\PntSet$ be a set of points in $\Re^d$, and $\PntSet_1$ and
    $\PntSet_2$ be two disjoint subsets of $\PntSet$. The sets
    $\PntSet_1$ and $\PntSet_2$ are \emphi{Voronoi separated} in
    $\PntSet$ if for all $\pnt_1 \in \PntSet_1$ and $\pnt_2 \in
    \PntSet_2$, we have that their Voronoi cells are disjoint; that
    is, $\VorCell{\pnt_1}{\PntSet} \cap \VorCell{\pnt_2}{\PntSet} =
    \emptyset$.
\end{defn}

\begin{defn}
    For a set $\PntSet$, a \emphi{partition} of $\PntSet$ is a pair of
    sets $\pth[]{\PntSet_1, \PntSet_2}$, such that $\PntSet_1
    \subseteq \PntSet$, and $\PntSet_2 = \PntSet \setminus \PntSet_1$.
    A set $\SepSet$ is a \emphi{Voronoi separator} for a partition
    $\pth[]{\PntSet_1, \PntSet_2}$ of $\PntSet \subseteq \Re^d$, if
    $\PntSet_1$ and $\PntSet_2$ are Voronoi separated in $\PntSet \cup
    \SepSet$; that is, the Voronoi cells of $\PntSet_1$ in
    $\VorX{\PntSet \cup \SepSet}$ do not intersect the Voronoi cells
    of $\PntSet_2$. We will refer to the points of the separator
    $\SepSet$ as \emphi{guards}.
\end{defn}

See \figref{v:partition} for an example of the above definitions.

\begin{defn}%
    \deflab{dbl:constant}%
    For a ball $\ballA$, its \emphi{covering number} is the minimum
    number of (closed) balls of half the radius that are needed to
    cover it.  The \emphi{doubling constant} of a metric space is the
    maximum cover number over all possible balls.  Let $\dblCd$ be the
    doubling constant for $\Re^d$.
\end{defn}

The constant $\dblCd$ is exponential in $d$, and $\dblCd \leq
\ceil{2\sqrt{d}}^d$ -- indeed, cover a ball (say, of unit radius) by a
grid with sidelength $1/\sqrt{d}$, and observe that each grid cell has
diameter $1$, and as such can be covered by a ball of radius $1/2$.

\begin{figure}[t]%
    \begin{tabular}{\si{cc}}
        \begin{minipage}{0.4\linewidth}
            \centerline{\includegraphics{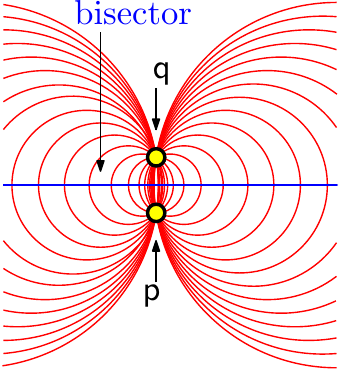}}%
        \end{minipage}
        &
        \begin{minipage}{0.4\linewidth}
            \centerline{\includegraphics[page=2]{figs/pencil}}%
        \end{minipage}
        \\
        (A) & (B)
    \end{tabular}
    \captionof{figure}{(A) The unbounded bisector induced by $\pnt$
       and $\pntA$. %
       (B) The pencil of $\pnt$ and $\pntA$. }%
    \figlab{hypercube}
\end{figure}%

\begin{defn}
    \deflab{projection}%
    For a closed set $\setA \subseteq \Re^d$, and a point $\pnt \in
    \Re^d$, the \emphi{projection} of $\pnt$ into $\setA$ is the
    closest point on $\setA$ to $\pnt$. We denote the projected point
    by $\projY{\pnt}{\setA}$.
\end{defn}


\section{Computing a small Voronoi separator}
\seclab{balanced:separator}

\subsection{Preliminaries, and how to block a sphere}

Given a set $\PntSet$ of $n$ points in $\Re^d$, we show how to compute
a balanced Voronoi separator for $\PntSet$ of size $O\pth{n^{1-1/d}}$.

\begin{defn}
    A set $\setB \subseteq \setA \subseteq \Re^d$ is
    \emphi{$\radN$-dense} in $\setA$, if for any point $\pnt \in
    \setA$, there exists a point $\pntB \in \setB$, such that
    $\dist{\pnt}{\pntB} \leq \radN$.
\end{defn}

\begin{lemma}
    \lemlab{blocker}%
    Consider an arbitrary sphere $\sphereB$, and a point $\pnt \in
    \Re^d \setminus \sphereB$. Then one can compute, in constant time,
    a set of points $\PntSetB \subseteq \sphereB$, such that the
    Voronoi cell $\VorCell{\pnt}{\PntSetB \cup \brc{\pnt}}$ does not
    intersect $\sphereB$, and $\cardin{\PntSetB} = O(1)$. We denote
    the set $\PntSetB$ by $\blockerY{\pnt}{\sphereB}$.
\end{lemma}

\begin{proof}
    If $\pnt$ is outside the sphere $\sphereB$, then $\PntSetB = \brc{
       \projY{\pnt}{\sphereB}}$ provides the desired separation.
    
    \parpic[r]{\includegraphics{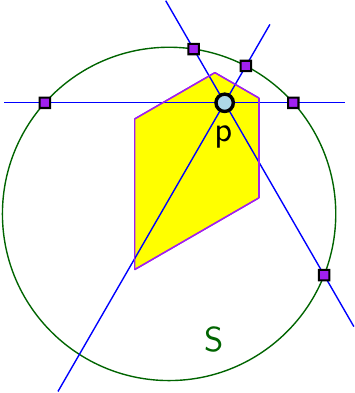}}
    
    If $\pnt$ is inside $\sphereB$, then consider the unit sphere
    centered at $\pnt$, cover it with the minimum number of spherical
    caps having diameter $\leq \pi/3$, and let $\CapSet$ be the
    resulting set of caps. Every such cap of directions defines a
    natural cone centered at $\pnt$. Formally, for such a cap $C$,
    consider the set $\coneY{\pnt}{C} = \brc{ \pnt + t \pntB \sep{
          \pntB \in C, t \geq 0}}$. Compute the closest point of
    $\sphereB$ inside this cone, and add the point to
    $\PntSetB$. Repeat this process for all the caps of $\CapSet$.
    
    We claim that $\PntSetB$ is the desired blocker. To this end,
    consider any cap $C\in \CapSet$, and observe that
    $\coneY{\pnt}{C}$ contains $\pntB \in \PntSetB$, and this is the
    closest point on $\sphereB \cap \coneY{\pnt}{C}$ to $\pnt$. In
    particular, since the cone angle is $\leq \pi/3$, it is
    straightforward to verify that the bisector of $\pnt$ and $\pntB$
    separates $\sphereB \cap \coneY{\pnt}{C}$ from $\pnt$, implying
    that $\VorCell{\pnt}{\PntSetB \cup \brc{\pnt}}$ can not intersect
    the portion of $\sphereB$ inside $\coneY{\pnt}{C}$, see figure
    above for an example.
\end{proof}

\subsection{The Algorithm}

The input is a set $\PntSet$ of $n$ points in $\Re^d$. The algorithms
works as follows:

\noindent
\begin{minipage}{0.68\linewidth}
    \begin{compactenum}[\noindent \;\; (A)]
        \item Let $c_d = \dblCd + 1$, see \defref{dbl:constant}.  Let
        $\ballX{\cenA}{\radA}$ be the smallest (closed) ball that
        contains $n/c_d$ points of $\PntSet$ where $\cenA\in\Re^d$.
        
        \noindent
        \item Pick a number $\radB$ uniformly at random from the range
        $\pbrcS{\radA, 2\radA}$.
        
        \item Let $\ballB = \ballX{\cenA}{\radB}$.
        
        \item Let $\PntSet_1 = \PntSet \cap \ballB$ and $\PntSet_2 =
        \PntSet \setminus \ballB$.
        
        \item Let $\radN=\radB / n^{1/d}$.  Compute an $\radN$-dense
        set $\SepSet$, of size $O\pth{\pth{ \radB / \radN }^{d-1}} =
        O\pth{ n^{1-1/d}}$, on the sphere $\sphereB = \partial \ballB$
        using the algorithm of \lemref{dense:set} described below.

        \item If a point $\pnt \in \PntSet$ is in distance smaller
        than $\radN$ from $\sphereB$, we insert
        $\blockerY{\pnt}{\sphereB}$ into the separating set $\SepSet$,
        see \lemref{blocker}.
        
    \end{compactenum}%
\end{minipage}%
\hfill%
\begin{minipage}{0.3\linewidth}
    \centerline{\includegraphics{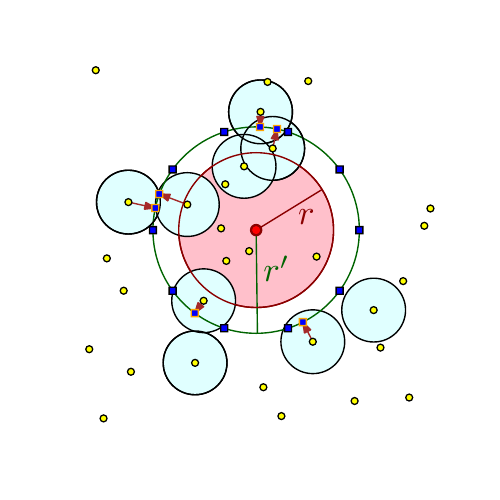}}
    \captionof{figure}{A slightly inaccurate depiction of how the
       algorithm works.}
\end{minipage}

\medskip

\noindent%
We claim that the resulting set $\SepSet$ is the desired separator.

\paragraph{Efficient implementation.}

One can find a $2$-approximation (in the radius) to the smallest ball
containing $n / c_d$ points in linear time, see \cite{h-gaa-11}. This
would slightly deteriorate the constants used above, but we ignore
this minor technicality for the sake of simplicity of exposition. If
the resulting separator is too large (i.e., larger than $\Omega\pth{
   n^{1-1/d}}$ see below for details), we rerun the algorithm.

\subsubsection{Computing a dense set}

The following is well known, and we include it only for the sake of
completeness.
\begin{lemma}
    \lemlab{dense:set}%
    Given a sphere $\sphereB$ of radius $\radB$ in $\Re^d$, and given
    a number $\radN>0$, one can compute a $\radN$-dense set $\setA$ on
    $\sphereB$ of size $O\pth{ \pth{\radB/\radN}^{d-1}}$. This set can
    be computed in $O\pth{ \cardin{\setA}}$ time.
\end{lemma}
\begin{proof}
    Consider the grid $\Grid$ of sidelength $\radN /\sqrt{d}$, and let
    $\setA$ be the set of intersection points of the lines of $\Grid$
    with $\sphereB$. Observe that every $(d-1)$-face of the bounding
    cube of $\sphereB$ intersects $O\pth{ \pth{\radB/\radN}^{d-1}}$
    lines of the grid, and since there $2d$ such faces, the claim on
    the size of $\setA$ follows.
    
    As for the density property, observe that for any point $\pnt \in
    \sphereB$, let $\Cell$ be the grid cell of $\Grid$ that contains
    it. Observe, that $\ballX{\pnt}{\radN}$ contains $\Cell$
    completely, one of the vertices of $\Cell$ must be inside the
    sphere, and at least one of them must be outside the sphere. Since
    the edges of the boundary of $\Cell$ are connected, it follows
    that one of the points of $\setA$ is on the boundary of $\Cell$,
    which in turn implies that there is a point of $\setA$ contained
    in $\ballX{\pnt}{\radN}$, implying the desired property.
\end{proof}

\subsection{Correctness}

\begin{lemma}
    We have $\cardin{\PntSet_1} \geq n/c_d$ and $\cardin{\PntSet_2}
    \geq n/c_d$.
\end{lemma}
\begin{proof}
    By \defref{dbl:constant}, the ball $\ballB = \ballX{\cenA}{\radB}$
    can be covered by $\dblCd$ balls of radius $\radA$, each one of
    them contains at most $n/ c_d$ points, as $\ballX{\cenA}{\radA}$
    is the smallest ball containing $n/ c_d$ points of $\PntSet$.
    
    As such $\ballB$ contains at most $\dblCd n/c_d$ points of
    $\PntSet$. In particular, as $c_d = \dblCd + 1$, we have that
    $\ballB$ has at least $n/c_d$ points of $\PntSet$, inside it, and
    at least $n(1-\dblCd/c_d)=n/c_d$ points outside it.
\end{proof}

\begin{lemma}
    The sets $\PntSet_1$ and $\PntSet_2$ are Voronoi separated in
    $\VorX{ \PntSet \cup \SepSet}$.
\end{lemma}

\begin{proof}
    We claim that all the points on $\sphereB$ are dominated by
    $\SepSet$. Formally, for any $\pntB \in \sphereB$, we have that
    $\distSet{\pntB}{\SepSet} \leq \distSet{\pnt}{\PntSet}$, which
    clearly implies the claim.
    
    So, let $\pntA$ be the nearest neighbor to $\pntB$ in $\PntSet$.
    If $\dist{\pntB}{\pntA} \geq \radN$ then since $\SepSet$ is
    $\radN$-dense in $\sphereB$, it follows that there exists $\pntC
    \in \SepSet$ such that $\distSet{\pntB}{\SepSet} \leq \dist{
       \pntB}{\pntC} \leq \radN \leq \dist{\pntB}{\pntA} =
    \distSet{\pntB}{\PntSet}$, as desired.
    
    If $\dist{\pntB}{\pntA} < \radN$ then the addition of
    $\blockerY{\pntA}{\sphereB}$ to $\SepSet$, during the
    construction, guarantees that the nearest point in $\SepSet$ to
    $\pntB$, is closer to $\pntB$ than $\pntA$ is, see
    \lemref{blocker}.
\end{proof}

\begin{lemma}
    \lemlab{small:separator}%
    Let $Y = \cardin{\SepSet}$. We have that $\Ex{Y} \leq \constSep
    n^{1-1/d}$, where $\constSep$ is some constant.
\end{lemma}
\begin{proof}
    Let $Z$ be the number of points of $\PntSet$, whose projections
    were added to $\SepSet$. We claim that $\Ex{Z} =
    O\pth{n^{1-1/d}}$.  To this end, for a point $\pnt \in \PntSet$,
    let $X_\pnt$ be the indicator variable that is one if and only if
    $\pnt$ is in distance at most $\radN$ from $\sphereB$. The
    algorithm picked the radius $\radB$ uniformly at random in the
    interval $[\radA, 2\radA]$. Furthermore, $X_\pnt= 1$ if and only
    if $\dist{\pnt}{\cenA} - \radN \leq \radB \leq \dist{\pnt}{\cenA}
    + \radN$. This happens only if $\radB$ falls into an interval
    $I_\pnt$ that is of length at most $2\radN$ centered at
    $\dist{\pnt}{\cenA}$.  As such, we have
    \begin{equation*}
        \Prob{\MakeBig X_{\pnt}=1} %
        =%
        \frac{ \lengthX{ \MakeSBig  I_\pnt \cap [r,2r] }}{\lengthX{
              \MakeSBig [r,2r]}} %
        \leq%
        \frac{2\radN}{\radA}%
        = %
        \frac{2\pth{\radB / n^{1/d}}}{\radA} \leq \frac{4 }{n^{1/d}}.
    \end{equation*}
    Now, by linearity of expectation, we have that $\Ex{Z} =
    \sum_{\pnt \in \PntSet} \Ex{c X_\pnt } \leq O( n^{1-1/d})$, where
    $c$ is the constant of \lemref{blocker}. This implies the claim,
    as $Y = Z + O\pth{n^{1-1/d} }$.
\end{proof}

\subsection{The result}


\begin{theorem}
    \thmlab{separator:main}%
    Let $\PntSet$ be a set of $n$ points in $\Re^d$. One can compute,
    in expected linear time, a sphere $\sphereB$, and a set $\SepSet
    \subseteq \sphereB$, such that
    \begin{compactenum}[\quad (i)]
        \item $\cardin{\SepSet} = O\pth{n^{1-1/d}}$,
        \item $\sphereB$ contains $\geq c n$ points of $\PntSet$
        inside it,
        \item there are $\geq c n$ points of $\PntSet$ outside
        $\sphereB$, and
        \item $\SepSet$ is a Voronoi separator of the points of
        $\PntSet$ inside $\sphereB$ from the points of $\PntSet$
        outside $\sphereB$.
    \end{compactenum}
    Here $c > 0$ is a constant that depends only on the dimension $d$.
\end{theorem}

\begin{proof}
    Clearly, each round of the algorithm takes $O(n)$ time. By
    Markov's inequality the resulting separator set $\SepSet$ is of
    size at most $2\constSep n^{1-1/d}$, with probability at least
    $1/2$, see \lemref{small:separator}. As such, if the separator is
    larger than this threshold, then we rerun the algorithm. Clearly,
    in expectation, after a constant number of iterations the
    algorithm would succeed, and terminates. (It is not hard to
    derandomize this algorithm and get a linear running time.)
\end{proof}


\section{Exact algorithm for computing the optimal %
   separation for a given partition}
\seclab{exact:separator}

Given a set $\PntSet$ of $n$ points in $\Re^d$, and a partition
$\pth[]{\PntSet_1,\PntSet_2}$ of $\PntSet$, we are interested in
computing the smallest Voronoi separating set realizing this
partition.

\subsection{Preliminaries and problem statement}
\seclab{euclidean:preliminaries}

\begin{defn}%
    \deflab{bad:pairs}%
    For a set $\PntSet\in \Re^d$ and a pair of disjoint subsets
    $(\PntSet_1,\PntSet_2)$, the set of \emphi{bad pairs} is
    \begin{math}
        \violatorX{\PntSet}{\PntSet_1}{\PntSet_2}%
        =%
        \brc{ (\pnt_1,\pnt_2) \in \PntSet_1 \times \PntSet_2 \sep{\,
              \VorCell{\pnt_1}{\PntSet} \cap \VorCell{\pnt_2}{\PntSet}
              \ne \emptyset}}.
    \end{math}
\end{defn}

\remove{%
   \begin{figure}[t]%
       \centerline{%
          \includegraphics{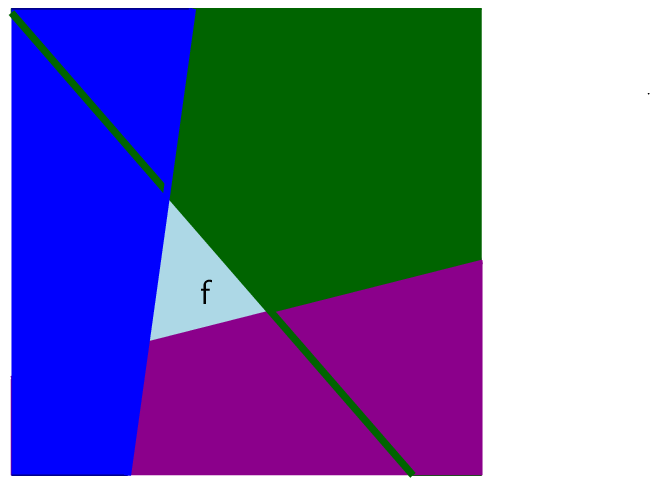}%
       }%
       \captionof{figure}{ A $2$-feature $\feature$ and its induced
          $2$-halfflats $h_1, h_2, h_3$. }
       \figlab{feature:halfspace}
   \end{figure}%
}

\parpic[r]{%
   \begin{minipage}{6cm}%
       \centerline{%
          {\includegraphics{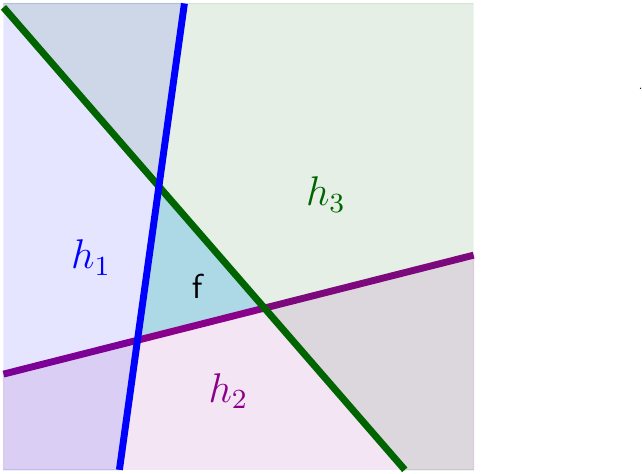}\hspace{-1.5cm}}%
       }%
       \captionof{figure}{ A $2$-feature $\feature$ and its induced
          $2$-halfflats $h_1, h_2, h_3$. }
       \figlab{feature:halfspace} \end{minipage}%
}

For a Voronoi diagram $\VorX{\PntSet}$, we can assume that all its
faces (of various dimensions) are all triangulated (say, using
bottom-vertex triangulation). This does not change the complexity of
the Voronoi diagram. For $k=0,1,\ldots, d$, such a $k$ dimensional
Voronoi simplex is a \emphi{$k$-feature}. Such a $k$-feature
$\feature$, is induced by $d-k+1$ sites, denoted by
$\sitesX{\feature}$; that is, any point $\pnt \in \feature$ is in
equal distance to all the points of $\sitesX{\feature}$ and these are
the nearest neighbor of $\pnt$ in $\PntSet$. Thus, a vertex $v$ of the
Voronoi diagram is a $0$-feature, and $\cardin{\sitesX{v}} = d+1$
(assuming general position, which we do).  The \emphi{span} of a
feature $\feature$, is the set of points in $\Re^d$ that are
equidistant to every site in $\sitesX{\feature}$; it is denoted by
$\spanF{\feature}$ and is the $k$ flat that contains $\feature$.  A
\emphi{$k$-halfflat} is the intersection of a halfspace with a
$k$-flat.

Consider any $k$-feature $\feature$. The complement set
$\spanF{\feature} \setminus \feature$ can be covered by $k+1$
$k$-halfflats of $\spanF{\feature}$. Specifically, each of these
halfflats is an open $k$-halfflat of $\spanF{\feature}$, whose
boundary contains a $(k-1)$-dimensional face of the boundary of
$\feature$. This set of halfflats of $\feature$, is the \emphi{shell}
of $\feature$, and is denoted by $\shellF{\feature}$, see
\figref{feature:halfspace}.

Once the Voronoi diagram is computed, it is easy to extract the ``bad
features''. Specifically, the set of \emphi{bad features} is
\begin{align*}
    \bFeaturesX{\PntSet}{\PntSet_1}{\PntSet_2}%
    =%
    \brc{ \feature \in \mathrm{features}\pth{\VorX{\PntSet}} \sep{
          \sitesX{\feature} \cap\PntSet_1 \neq \emptyset \text{\;
             and\; } \sitesX{\feature} \cap\PntSet_2 \neq \emptyset}}.
\end{align*}
Clearly, given a Voronoi diagram the set of bad features can be
computed in linear time in the size of the diagram.  \medskip

Given a $k$-feature $\feature$, it is the convex-hull of $k+1$ points;
that is, $\feature = \CHX{\setA}$, where $\setA = \brc{\pntA_1,
   \ldots, \pntA_{k+1}}$. We are interested in finding the closest
point in a feature to an arbitrary point $\pnt$. This is a constant
size problem for a fixed $d$, and can be solved in constant time.  We
denote this closest point by $\nnX{\pnt}{\feature} = \arg \min_{\pntA
   \in \feature}\distSet{\pnt}{\pntA}$.  For the feature $\feature$,
and any point $\pnt$, we denote by $\pencilPnt{\feature}{\pnt}$ the
ball $\ballX{\pnt}{\distSet{\pnt}{\sitesX{\feature}}}$ (if it is
uniquely defined).  Furthermore, for an arbitrary set $S$ of points,
$\pencilSet{\feature} {S}$ denote
$\brc{\ballX{\pnt}{\distSet{\pnt}{\sitesX{\feature}}}\sep{ \pnt\in
      S}}$. In particular, for any point $\pnt \in \feature$, consider
$\ballX{\pnt}{\distSet{\pnt}{\PntSet}}$ -- it contains the points of
$\sitesX{\feature}$ on its boundary. The set of all such balls is the
\emphi{pencil} of $\feature$, denoted by
\begin{align}
    \pencilF{\feature} = \brc{
       \ballX{\pnt}{\distSet{\pnt}{\sitesX{\feature}}} \sep{ \pnt \in
          \feature}}.%
    \eqlab{pencil}%
\end{align}
The \emphi{trail} of $\feature$ is the union of all these balls; that
is, $\trailX{\feature} = \bigcup_{\pnt \in \feature}
\ballX{\pnt}{\distSet{\pnt}{\PntSet}}$. Finally, let $\mbX{\feature}$
denote the smallest ball in the pencil of a feature
$\feature$. Clearly, the center of $\mbX{\feature}$ is the point
$\nnX{\pnt}{\feature}$, where $\pnt$ is some arbitrary point of
$\sitesX{\feature}$. As such, $\mbX{\feature}$ can be computed in
constant time.

\begin{lemma}%
    \lemlab{halfspace}%
    Let $\pnt$ be any point and let $\feature$ be any $k$-feature. The
    point $\pnt$ induces a halfflat of $\spanF{\feature}$ denoted by
    $\halfspaceX{\pnt} {\feature}$, such that
    $\pencilF{\halfspaceX{\pnt} {\feature}}$ is the set of all balls
    in $\pencilF{ \spanF{\feature}}$ that contain $\pnt$.
\end{lemma}

\begin{proof}
    Consider any arbitrary site $\pntB\in \sitesX{\feature}$.  The set
    of points whose ball in the pencil contains $\pnt$, is clearly the
    set of points in $\spanF {\feature}$ that are closer to $\pnt$
    than to $\pntB$.  This set of points, is a halfspace of $\Re^d$
    that is not parallel to the $k$-flat $\spanF{\feature}$.  Thus,
    $\halfspaceX{\pnt}{\feature}=\brc{\pntA\sep{
          \dist{\pntA}{\pnt}\leq \dist{\pntA}{\pntB}}}\cap
    \spanF{\feature}$ is the desired halfflat of $\spanF{\feature}$
    induced by $\pnt$, whose boundary is given by the set of points
    equidistant to $\sitesX {\feature}\cup \brc{\pnt}$.
\end{proof}


We are now ready to restate our problem in a more familiar language.

\begin{lemma}[Restatement of problem]%
    \lemlab{hit:balls}%
    Given a set $\PntSet$ of $n$ points in $\Re^d$, and a pair of
    disjoint subsets $\pth[]{\PntSet_1,\PntSet_2}$, finding a minimum
    size Voronoi separator realizing separation of
    $\pth[]{\PntSet_1,\PntSet_2}$, is equivalent to finding a minimum
    size hitting set of points $\SepSet$, such that $\SepSet$ stabs
    (the interior) of all the balls in the set
    \begin{align}
        \BadBallsA%
        =%
        \BadBalls{\PntSet}{\PntSet_1}{\PntSet_2} = \bigcup_{\feature
           \in \bFeaturesX{\PntSet}{\PntSet_1}{\PntSet_2}}
        \pencilF{\feature}.%
        \eqlab{hitting:set}
    \end{align}
\end{lemma}
\begin{proof}
    Indeed, a Voronoi separating set $\SepSet$, must stab all the
    balls of $\BadBallsA$, otherwise, there would be Voronoi feature
    of $\VorX{ \PntSet \cup \SepSet}$ that has a generating site in
    both, $\PntSet_1$ and $\PntSet_2$.
    
    As for the other direction, consider a set $\setA$ that stabs all
    the balls of $\BadBallsA$, and observe that if $\PntSet_1$ and
    $\PntSet_2$ are not Voronoi separated in $\VorX{ \PntSet \cup
       \setA}$, then there exists a ball $\ballA$, that has no point
    of $\PntSet \cup \setA$ in its interior, and points from both
    $\PntSet_1$ and $\PntSet_2$ on its boundary. But then, this ball
    must be in the pencil of the Voronoi feature induced by $\PntSet
    \cap \ballA$. A contradiction to $\setA$ stabbing all the balls of
    $\BadBallsA$.
\end{proof}

\subsection{Exact algorithm in $\Re^d$}
\seclab{exact}

Given a set $\PntSet$ of $n$ points in $\Re^d$, a pair of disjoint
subsets $(\PntSet_1,\PntSet_2)$ of $\PntSet$ and an upper bound
$\maxGuard$ on the number of guards, we show how one can compute the
minimum size Voronoi separator realizing their separation in
$n^{O(\maxGuard)}$ time. Our approach is to construct a small number
of polynomial inequalities that are necessary and sufficient
conditions for separation, and then use cylindrical algebraic
decomposition to find a feasible solution.

\begin{observation}
    \obslab{feature:cover} For a set of guards $\SepSet$, a feature
    $\feature$ is completely removed from $\VorX{\PntSet\cup\SepSet}$,
    if and only if the induced halfflats of the guards, on
    $\spanF{\feature}$, cover $\feature$. This follows directly from
    \lemref{hit:balls} and \lemref{halfspace}.
\end{observation}

\begin{observation}
    \obslab{helly}%
    Given a set $\HH$ of at least $k+1$ halfflats in a $k$-flat, if
    $\HH$ covers the $k$-flat, then there exists a subset of
    $(k+1)$-halfflats of $\HH$, that covers the $k$-flat.  This is a
    direct consequence of Helly's theorem.
\end{observation}

\paragraph{Low dimensional example.}
To get a better understanding of the problem at hand, the reader may
imagine the subproblem of removing a bad $2$-feature $\feature$
(i.e. $\feature$ is a triangle) from the Voronoi diagram. We know that
a set $\SepSet$ of $n$ guards removes $\feature$ from the Voronoi
diagram, if and only if the $n$ $2$-halfflats induced by $\SepSet$
cover the triangle $\feature$. If we add the feature-induced halfflats
induced by $\feature$ to the above set of halfflats, then the problem
of covering the triangle $\feature$, reduces to the problem of
covering the entire plane with this new set $S$ of $n+3$
halfflats. Then from \obsref{helly}, we have that $\feature$ is
removed from the Voronoi diagram if and only if there are three
halfplanes of $S$, that covers the entire plane (there are $O(n^3)$
such triplets).  \lemref{annihilate:feature} below show how to convert
the condition that any three $2$-halfflats cover the plane, into a
polynomial inequality of degree four in the coordinates of the guards.

\subsubsection{Constructing the Conditions}


\begin{lemma}
    \lemlab{feature:cover}
    Let $\feature$ be a $k$-feature, and $\HH$ be a set of
    $k$-halfflats on $\spanF{\feature}$.  Then, $\HH$ covers
    $\feature$ $\iff$ there exists a subset $\HHB \subseteq \HH' = \HH
    \cup \shellF{\feature}$ of size $k+1$ that covers
    $\spanF{\feature}$.
\end{lemma}

\begin{proof}    
    If $\HH$ covers $\feature$, then $\HH' = \HH \cup
    \shellF{\feature}$ covers $\spanF{\feature}$, as
    $\shellF{\feature}$ covers
    $\spanF{\feature}\setminus\feature$. Then, by Helly's theorem (see
    \obsref{helly}), we have that some subset of $\HH'$ of size $k+1$
    covers $\spanF{\feature}$.
    
    For the other direction, we have that some subset of $\HH'$ covers
    $\spanF{\feature}$ (of size $k+1$).  Since $\shellF{\feature}$
    does not cover any point in $\feature$, we have that $\HH = \HH'
    \setminus \shellF{\feature}$ must cover $\feature$ completely.
\end{proof}

\begin{observation}
    \obslab{plane:cover}%
    Consider a set $\HH$ of $k+1$ $k$-halfflats all contained in some
    $k$-flat $\Flat$. We are interested in checking that $\HH$ covers
    $\Flat$. Fortunately, this can be done by computing the $k+1$
    vertices induced by $\HH$ on $\Flat$, and verifying that each one
    of them is covered by the other halfflat of $\HH$. Formally, $\HH$
    covers $\Flat$ $\iff$ for every halfflat $h\in \HH$, we have that
    ${\cap_{i\in \HH\setminus \{h\}}\partial i} \,\in\, h$.
\end{observation}

For a set $\PntSet$ of $d+1$ points in $\Re^{d}$ in general position,
let $\inducedBall{ \PntSet }$ be the unique ball having all the points
of $\PntSet$ on its boundary.

\begin{lemma}
    \lemlab{incircle}%
    Consider a $k$-feature $\feature$, let $\HH$ be a set of halfflats
    on $\spanF{\feature}$, and let $\SepSet$ be the set of guards
    inducing the halfflats of $\HH$. Assume that $\cardin{\sitesX{
          \feature }} + \cardin{\SepSet} = d+2$.  Let $h \in \HH$ be a
    $k$-halfflat induced by a guard $\grd$, and let $\pnt =
    \bigcap_{i\in \HH \setminus \brc{h}}\partial i$.  Then $\pnt \in
    h$ $\iff$ $\grd \in \inducedBall{ \MakeBig \sitesX {\feature}\cup
       \pth[]{ \MakesBig \SepSet \setminus \brc{\grd} } }$.
\end{lemma}

\begin{proof}
    The boundary of the induced halfflat of a guard $\grd$, is the set
    of points equidistant to the points of $\sitesX{ \feature } \cup
    \grd$, see \lemref{halfspace}; equivalently, this is the set
    containing every point $\pnt$, such that $\pencilPnt{ \feature }{
       \pnt }$ contains $\grd$ on its boundary. Furthermore, for any
    $\pntA$ inside this induced halfflat, $\grd \in \pencilPnt{
       \feature }{ \pntA }$. Thus $\grd \in \pencilPnt{ \feature }{
       \pnt }$ $\iff$ the point $\pnt \in h$. Also, $\pnt$ is the
    point in equal distance to all the points of
    $\sitesX{\feature}\cup \SepSet \setminus \brc{\grd}$.
\end{proof}

A set $\SepSet$ of $m$ guards in $\Re^d$ can be interpreted as a
vector in $\Re^{dm}$ encoding the locations of the guards. One can
then reduce the requirement that $\SepSet$ provides the desired
separation into a logical formula over the coordinates of this vector.

\begin{lemma}
    \lemlab{annihilate:feature}%
    Let $\feature$ be a bad $k$-feature of $\bFeaturesX
    {\PntSet}{\PntSet_1}{\PntSet_2}$, and let $m$ be a parameter. One
    can compute a boolean sentence $\annihilateF{\feature}{\SepSet}$
    consisting of $m^{O\pth[]{k}}$ polynomial inequalities (of degree
    $\leq d+2$), over $d m$ variables, such that
    $\annihilateF{\feature}{\SepSet}$ is true $\iff$ the set of $m$
    guards $\SepSet$ (induced by the solution of this formula)
    destroys $\feature$ completely when inserted. Formally, for every
    $\feature'\in \bFeaturesX{\PntSet\cup\SepSet}{\PntSet_1}
    {\PntSet_2}$, we have $\feature\cap\feature'= \emptyset$.
\end{lemma}

\begin{proof}
    Let $\HH=\brc{\halfspaceX{\pnt}{\feature}\sep {\pnt\in\SepSet}}$
    be the set of halfflats on $\spanF{\feature}$ induced by the
    guards in $\SepSet$. By \obsref{feature:cover}, if $\HH$ covers
    $\feature$, then $\SepSet$ removes $\feature$ from the Voronoi
    diagram (see also \lemref{hit:balls}). To this end, let $\HH' =
    \HH \cup \shellF{\feature}$. Further, let $\SepSet_\feature$ be a
    set of $k+1$ points that exactly induce the halfflats of $\shellF{
       \feature }$ and let $\SepSet'=\SepSet\cup\SepSet_
    \feature$. From \lemref{feature:cover}, we know that $\HH$ covers
    $\feature$ $\iff$ some size $k+1$ subset of $\HH'$ covers the
    $k$-flat $\spanF{\feature}$. Then by \obsref{plane:cover}, $\HH$
    covers $\feature$ $\iff$ there exists some size $k+1$ subset
    $\HH_1$ of $\HH'$, such that for every $h\in \HH_1$, we have
    $\bigcap_{i\in \HH_1 \setminus h}\partial i\in h$.
    
    Since $\SepSet'$ exactly induces the halfflats of $\HH'$,
    \lemref{incircle} implies that $\feature$ is completely removed
    from $\VorX {\PntSet\cup\SepSet}$ $\iff$ there exists some size
    $k+1$ subset $\SepSet_1$ of $\SepSet'$, such that for every
    $\grd\in \SepSet_1$, $\grd$ lies inside the ball containing the
    points of $\sitesX{\feature}\cup \pth{\SepSet_1\setminus\grd}$ on
    its boundary. More formally, we define,
    \begin{align*}
        \annihilateF{\feature}{\SepSet}%
        \qquad\equiv \qquad%
        \bigvee _{\SepSet_1\subseteq\SepSet',\cardin{\SepSet_1}=k+1}
        \pth{ \bigwedge_{\grd\in\SepSet_1} \pbrc{\indet{ \MakeBig
                 \sitesX{\feature}\cup
                 \pth{\SepSet_1\setminus\grd}}{\grd} \leq 0 }},
    \end{align*}
    where $\indet{S}{\pnt}$ represents the standard in-circle
    determinant whose sign determines the inclusion of $\pnt$ in the
    ball fixed by the points of $S$. This determinant affords us a
    polynomial of degree at most $d+2$ involving only the variables of
    $\SepSet$, as the points of $\SepSet_\feature$ are merely
    constants.  Clearly, $\annihilateF{\feature}{\SepSet}$ contains at
    most $\cardin{\SepSet}^{O\pth[]{k}}$ polynomial inequalities, and
    the claim follows.
\end{proof}

\subsubsection{The Result}

\begin{theorem}
    \thmlab{exact}%
    Let $\PntSet$ be a set of $n$ points in $\Re^d$. Let
    $(\PntSet_1,\PntSet_2)$ be some disjoint partition of $\PntSet$
    and $\maxGuard$ be a parameter, such that there exists a Voronoi
    separator for $(\PntSet_1,\PntSet_2)$ of at most $\maxGuard$
    points. The minimal size Voronoi separator can be computed in
    $n^{O\pth[]{\maxGuard}}$ time, if such a separator exists of size
    at most $\maxGuard$ (note, that the constant in the $O$ notations
    depends on $d$).
\end{theorem}

\begin{proof}
    Let $\nGuard$ be a parameter to be fixed shortly. Let $\SepSet$ be
    a set of $\nGuard$ guards in $\Re^d$.  By
    \lemref{annihilate:feature}, the condition $\annihilateX{\SepSet}$
    that every bad feature is removed from the Voronoi diagram of
    $\PntSet\cup\SepSet$, can be written as
    \begin{align*}
        \annihilateX{\SepSet} = \bigwedge_{\feature \in
           \bFeaturesX[]{\PntSet}{\PntSet_1}{\PntSet_2}}
        \annihilateF{\feature}{\SepSet}.
    \end{align*}
    Namely, $\SepSet$ is a Voronoi separator for
    $(\PntSet_1,\PntSet_2)$.
    The formula $\annihilateX {\SepSet}$ contains $n^{O\pth[]{d}}$
    degree-$d+2$ polynomial inequalities comprising of at most $d
    \nGuard$ variables. By \cite[Theorem~13.12]{bpr-arag-06}, one can
    compute, in $n^{O\pth[]{\nGuard}}$ time, a solution as well as the
    sign of each polynomial in $\PP$, for every possible sign
    decomposition that $\PP$ can attain. Now for each attainable sign
    decomposition, we simply check if $\annihilateX{\SepSet}$ is
    true. This can be done in $n^{O\pth[]{d}}$ time. The algorithm
    computes a Voronoi separator for $(\PntSet_1,\PntSet_2)$ in
    $\PntSet$ of size $\nGuard$, in $n^{O\pth[]{\nGuard}}$ time, if
    such a separator exists.
    
    Now, the algorithms tries $\nGuard = 1, \ldots, \maxGuard$ and
    stops as soon as a feasible solution is found.
\end{proof}

\paragraph{Remark.}

Interestingly, both here and in the rest of the paper, one can specify
which pair of points need to be separated from each other (instead of
a partition of the input set), and the algorithm works verbatim for
this case (of course, in the worst case, one might need to specify
$O(n^2)$ pairs, if one explicitly lists the forbidden pairs).

\section{Approximation algorithms}

\subsection{Constant factor approximation}
\seclab{const:factor}

Given a set $\PntSet$ of $n$ points in $\Re^d$, and a partition
$(\PntSet_1,\PntSet_2)$ of $\PntSet$, we show how one can compute in
$n^{O(d)}$ time, a Voronoi separator $\SepSet$ for $(\PntSet_1,
\PntSet_2)$, whose size is a constant factor approximation to the size
of the minimal Voronoi separator realizing such a partition.

\subsubsection{The algorithm \GreedySeparator{}}
Since the problem is equivalent to a hitting set problem on balls, see
\lemref{hit:balls}, we can apply the ``standard'' greedy
strategy. Specifically, given $\PntSet \subseteq \Re^d$, and a pair of
disjoint subsets $\pth[]{\PntSet_1, \PntSet_2}$, start with an empty
set $\setA$ of stabbing points.

Now, at each iteration, compute the Voronoi diagram $\VorX{ \PntSet
   \cup \setA}$, and compute the set of bad features
$\bFeaturesX{\PntSet \cup \setA}{\PntSet_1}{\PntSet_2}$. Compute the
minimum ball in $\pencilF{\feature}$ (see \Eqref{pencil}) for each of
the bad features $\feature \in \bFeaturesX{\PntSet \cup
   \setA}{\PntSet_1}{\PntSet_2}$.  Let $\ballA=\ballX{\pnt}{r}$ be the
smallest ball encountered. Let $\PntSet_\ballA$ be a constant size set
of points, such that any ball that intersects $\ballA$ (and is at
least as large), must intersect $\PntSet_\ballA$. Such a set can be
easily constructed by sprinkling a $r$-dense set in the enlarged ball
$\ballX{\pnt}{2r}$. Note that $\PntSet_\ballA$ must also stab
$\ballA$.  The resulting set $\PntSet_\ballA$ has constant size (for
example, by using grid points of the appropriate size). Set $\setA
\leftarrow \setA \cup \PntSet_\ballA$, and repeat.

\subsubsection{The Result}

\begin{theorem}%
    \thmlab{const:sep}%
    Let $\PntSet$ be a set of $n$ points in $\Re^d$, and
    $\pth[]{\PntSet_1,\PntSet_2}$ be a given pair of disjoint subsets
    of $\PntSet$. One can compute a Voronoi separator $\SepSet$ that
    realizes this partition. The algorithm runs in $O(n^2 \log n)$
    time for $d=2$, and in $O\pth{n^{\ceil{d/2}+1}}$ time, for $d >
    2$.  The algorithm provides a constant factor approximation to the
    smallest Voronoi separator realizing this partition.
\end{theorem}

\begin{proof}
    Observe the optimal solution is finite, as one can place $d+1$
    points around each point of $\PntSet_1$, close enough, so that
    they Voronoi separate this point from the rest of $\PntSet$. As
    such, the optimal solution is of size at most
    $(d+1)\min\pth{\cardin{\PntSet_1}, \cardin{\PntSet_2}}$.
    
    As for the quality of approximation, this is a standard exchange
    argument. Indeed, consider the optimal solution, and observe that,
    at each iteration, the minimum ball found $\ballA$ must contain at
    least one point of the optimal solution, and the greedy algorithm
    replaces it in the approximation by the set $\PntSet_\ballA$, such
    that any ball that is stabbed by a point of the optimal solution
    that is contained in $\ballA$, and is at least as big as $\ballA$
    is also stabbed by some points of $\PntSet_\ballA$. In particular,
    if $c$ is the maximum size set $\PntSet_\ballA$ constructed by the
    algorithm (which is a constant), then the quality of approximation
    of the algorithm is $c$.
    
    The above implies that the algorithms terminates with at most
    $O(n)$ points added. Since computing the Voronoi diagram takes
    $O(n \log n)$ time in the plane, this implies that the running of
    the algorithm in the plane is $O(n^2 \log n)$.  In higher
    dimensions, computing the Voronoi diagram takes
    $O\pth{n^{\ceil{d/2}}}$, and overall, the algorithm takes
    $O\pth{n^{\ceil{d/2}+1}}$ time.
\end{proof}

\subsection{Polynomial time approximation scheme}
\seclab{PTAS}

Given a set $\PntSet$ of $n$ points in $\Re^d$, a parameter $\eps>0$,
and a pair of disjoint subsets $(\PntSet_1,\PntSet_2)$ of $\PntSet$,
we show here how to $(1+\eps)$-approximate, in $n^{O( 1/\eps^d)}$
time, a Voronoi separator $\SepSet$ for $(\PntSet_1,\PntSet_2)$.

As implied by \lemref{hit:balls}, this problem boils down to a hitting
set problem. That is, the desired partition give rise to a set of
balls $\SetBalls$, such that our task is to approximate the minimum
size hitting set for $\SetBalls$. The challenge here is that
$\SetBalls$ is an infinite set, see \Eqrefpage{hitting:set}.

We need an improved separator theorem, as follows.
\begin{theorem}
    \thmlab{ball:separator}%
    Let $\SetX$ be a set of points in $\Re^d$, and $k>0$ be an integer
    sufficiently smaller than $\cardin{\SetX}$.  One can compute, in
    $O\pth{\cardin{\SetX}}$ expected time, a set $\SepSet$ of
    $O\pth{k^{1-1/d}}$ points and a sphere $\sphereB$ containing
    $\Theta(k)$ points of $\SetX$ inside it, such that for any set
    $\SetBalls$ of balls stabbed by $\SetX$, we have that every ball
    of $\SetBalls$ that intersects $\sphereB$ is stabbed by a point of
    $\SepSet$.
\end{theorem}

Due to space limitations, we had delegated the proof of
\thmref{ball:separator} to \apndref{sperating:guards} -- it is a
careful extension of the proof of \thmref{separator:main}.

\subsubsection{Preliminaries}

For a set $\SetBalls$ of balls in $\Re^d$, let
$\hitSetSizeX{\SetBalls}$ be the size of its minimal hitting set.
Here, our aim is to approximate the value of
$\hitSetSizeX{\BadBalls{\PntSet}{\PntSet_1}{\PntSet_2}}$.

\begin{observation}
    \obslab{subset}%
    Let $\SetBalls$ and $\SetBalls'$ be sets of balls such that
    $\SetBalls' \subseteq \SetBalls$.  Then $\hitSetSizeX{\SetBalls'}
    \leq \hitSetSizeX{\SetBalls}$.
\end{observation}

\begin{observation}
    \obslab{disjoint:union}%
    Let $\SetBalls$ and $\SetBalls'$ be sets of balls such that no
    ball in $\SetBalls$ intersects a ball in $\SetBalls'$.  Then
    $\hitSetSizeX{\SetBalls \cup \SetBalls' } =
    \hitSetSizeX{\SetBalls} + \hitSetSizeX{ \SetBalls' }$.
\end{observation}

\subsubsection{The Algorithm \Algorithm{ApproximateSeparator}}

The input is a set $\PntSet$ of $n$ points in $\Re^d$, and a desired
partition $\pth[]{\PntSet_1, \PntSet_2}$.
\begin{compactenum}[\qquad (A)]
    \item Using the algorithm \GreedySeparator{}, compute a Voronoi
    separator $\SetX$, for $(\PntSet_1,\PntSet_2)$ in $\PntSet$, and
    let $k = O\pth{ 1/\eps^{d}}$.
    
    \item $i\leftarrow 1$
    
    \item $\SetX_1 \leftarrow \SetX$, $\SepSet_1 = \emptyset$
    
    \item While $\SetX_i \neq \emptyset$:

    \begin{compactenum}[(a)]
        \item Compute a sphere $\sphereB_i$ and a set $\oSepSet_i$
        using \thmref{ball:separator} for $\SetX_i$ with the parameter
        $\min\pth{ k, \cardin{ \SetX_i}}$.
        
        \item Let $\ballA_i$ be the corresponding closed ball of
        $\sphereB_i$.
        
        \item $\PntSetB_i \leftarrow
        (\PntSet\cap\ballA_i)\cup\SepSet_i \cup \oSepSet_i$, %
        \quad $\PntSet_{i,1} = \ballA_i \cap \PntSet_1$, %
        \quad $\PntSet_{i,2} = \ballA_i \cap \PntSet_2$, %
        \quad $\SetXX_i = \ballA_i \cap \SetX_i$.
        
        \item $\SepSet_i' \leftarrow $ Compute (exactly) the smallest
        hitting set for $\BadBalls{ \PntSetB_i }{ \PntSet_{i,1} }{
           \PntSet_{i,2} }$ using \thmref{exact}.
        
        \item $\SetX_{i+1} \leftarrow \SetX_i\setminus\ballA_i$
        
        \item $\SepSet_{i+1} \leftarrow \SepSet_i \cup
        \oSepSet_i\cup\SepSet_i'$
        
        \item $i\leftarrow i+1$
    \end{compactenum}
    
    \item Output $\SepSet_i$
    
\end{compactenum}

\subsubsection{Correctness}

Let $t$ be the number of iterations in the algorithm above.  Let
$\SetBalls = \BadBalls{\PntSet}{\PntSet_1}{\PntSet_2}$ be the set of
all balls that need to be hit (see \Eqrefpage{hitting:set}). For $i =
1, \ldots, t$, let $\SetBalls_i = \BadBalls{ \PntSetB_i }{
   \PntSet_{i,1} }{ \PntSet_{i,2} }$.

\begin{lemma}
    \lemlab{valid}%
    The set $\SepSet_t$ is a hitting set for $\SetBalls$.
\end{lemma}

\begin{proof}
    Observe that $\SetXX_1, \ldots, \SetXX_t$ form a partition of
    $\SetX$, and consider any ball $\ballB\in\SetBalls$.  The ball
    $\ballB$ is stabbed by at least one point of $\SetX$, and let $j$
    be the minimal index such that $\SetXX_j$ stabs $\ballB$.
    
    If $\ballB$ intersects the sphere $\sphereB_j$ then, $\oSepSet_j
    \subseteq \SepSet_t$ stabs $\ballB$, by \thmref{ball:separator}.
    Otherwise, $\ballB$ must be inside $\sphereB_j$, as otherwise it
    cannot be stabbed by the points of $\SetXX_j \subseteq
    \ballA_j$. But then, $\ballB$ being a ``bad'' ball, must have a
    point of $\PntSet_{j,1}$ and a point of $\PntSet_{j,2}$ on its
    boundary, implying $\ballB \in \SetBalls_j = \BadBalls{\PntSetB_j}
    {\PntSet_{j,1}} {\PntSet_{j,2}}$, and it is thus stabbed by
    $\SepSet_j' \subseteq \SepSet_t$, as testified by \thmref{exact}.
\end{proof}

\begin{lemma}
    \lemlab{geometrically:disjoint}%
    \begin{math}
        \hitSetSizeX{\bigcup_{i=1}^t \SetBalls_i} = \sum_ {i=1}^t
        \hitSetSizeX{\SetBalls_i}.
    \end{math}
\end{lemma}

\begin{proof}
    Consider any integers $i,j$ such that $1\leq i<j\leq t$.  Consider
    any ball $\ballA\in\SetBalls_i$. By construction, $\SepSet_i'$
    stabs $\ballA$. Now $\SepSet_i'\subseteq\SepSet_j
    \subseteq\PntSetB_j$, implying that $\ballA$ is stabbed by
    $\PntSetB_j$. Thus, $\ballA$ cannot be in the set $\SetBalls_i$,
    and further, $\SetBalls_i\cap\SetBalls_j = \emptyset$.
    
    A stronger property holds -- there is no point that stabs balls
    that are in both $\SetBalls_i$ and $\SetBalls_j$. Indeed, all the
    balls of $\SetBalls_i$ (resp.  $\SetBalls_j$) are contained inside
    $\sphereB_i$ (resp. $\sphereB_j$) by \thmref{ball:separator}. As
    such, such a point $\pnt$ that stabs balls in both sets, must be
    in $\ballA_i \cap \ballA_j$.  But $\SepSet_{i+1}$ stabs all the
    balls of $\SetBalls$ that intersect $\ballA_i$. In particular, as
    $\SepSet_{i+1} \subseteq \PntSetB_j$, it follows that no ball of
    $\SetBalls_j$ can intersect $\ballA_i$.
    
    Now \obsref{disjoint:union} implies the claim.
\end{proof}

\begin{lemma}
    \lemlab{approx}%
    Let $m$ be the size of the optimal hitting set.  We have
    $\cardin{\SepSet_t}\leq m\pth{1+O\pth{1/k^{1/d}}}$.
\end{lemma}

\begin{proof}
    We have $t \leq \ceil{\cardin{\SetX}/k } +1 =O(m/k)$, as $\SetX$
    is a constant approximation to the optimal hitting set. Then, by
    \lemref{geometrically:disjoint} and \obsref{subset}, we have
    \begin{align*}
        \cardin{\SepSet_t}%
        &\leq%
        \sum_{i=1}^t\pth{\cardin{\SepSet_i'} + \cardin{\oSepSet_i}}%
        \leq%
        \sum_{i=1}^t\cardin{\SepSet_i'} + t \cdot O\pth{k^{1-1/d}}%
        \leq%
        \sum_{i=1}^t \hitSetSizeX{\SetBalls_i} +
        O\pth{\frac{m}{k^{1/d}}}%
        \\%
        &=%
        \hitSetSizeX{\bigcup_{i=1}^t \SetBalls_i} +
        O\pth{\frac{m}{k^{1/d}}}%
        \leq%
        \hitSetSizeX{\SetBalls} + O\pth{\frac{m}{k^{1/d}}}%
        =%
        m\pth{ 1+O\pth{\frac{1}{k^{1/d}}} } .
    \end{align*}
\end{proof}

\subsubsection{The result}

\begin{theorem}
    Given a set $\PntSet$ of $n$ points in $\Re^d$, a parameter
    $\eps>0$, and a corresponding partition $(\PntSet_1,\PntSet_2)$,
    one can compute a Voronoi separator for $(\PntSet_1,\PntSet_2)$ in
    $n^{O\pth[]{1/\eps^d}}$ time, such that its size is a
    $(1+\eps)$-approximation to the minimal Voronoi separator
    realizing the partition $(\PntSet_1,\PntSet_2)$.
\end{theorem}

\begin{proof}
    By \lemref{valid} and \lemref{approx}, setting $k = c / \eps^d$,
    where $c$ is a sufficiently large constant implies the desired
    approximation.
    
    As for the running time, observe that the bottleneck in the
    running time, is in the invocation of the exact algorithm of
    \thmref{exact} in the $i$\th iteration, for $i=1,\ldots,
    t$. However, in the $i$\th iteration, a trivial stabbing set is
    $\SetXX_i$, which has size $O(k)$. That is, the running time of
    algorithm of \thmref{exact} in the $i$\th iteration is $n^{O(k)}$,
    implying the result.
\end{proof}

\section{\PTAS for geometric hitting set %
   via local search}
\seclab{local:search}

Here, we present a simple local search $(1+\eps)$-approximation
algorithm (\PTAS) for hitting set problems of balls (or fat objects)
in $\Re^d$. The set of balls is either specified explicitly, or as in
the case of the Voronoi partition, implicitly.

\subsection{The Algorithm \Algorithm{LocalHitBalls}}

\subsubsection{Preliminaries}

Let $\SetBalls$ be a set of balls (possibly infinite) in $\Re^d$ that
is represented (potentially implicitly) by an input of size $n$, and
let $\hitSetSizeX{\SetBalls}$ be the size of its minimal hitting set.
We assume that $\SetBalls$ satisfies the following properties:
\begin{compactenum}[\qquad(P1)]
    \item \pitemlab{p:1}%
    \textbf{Initial solution}: One can compute a hitting set $\SetX_0$
    of $\SetBalls$ in $n^{O\pth[]{1}}$ time, such that the size of
    $\SetX_0$ is a constant factor approximation to the optimal.
    
    \item \pitemlab{p:2}%
    \textbf{Local exchange}: Let $\SetX$ and $\SetY$ be point
    sets, such that
    \begin{compactenum}[\qquad (i)]
        \item $\SetX$ is a hitting set of $\SetBalls$ (i.e., it is the
        current solution),
        \item $\cardin{\SetX} \leq n^{O\pth[]{1}}$ (i.e., it is not
        too large),
        \item $\SetY\subseteq\SetX$ (i.e., $Y$ is the subset to be
        replaced),
        \item $\cardin{\SetY} \leq \lSize$, where $\lSize$ is any
        integer.
    \end{compactenum}
    Then one can compute in $n^{O\pth[]{\lSize}}$ time, the smallest
    set $\SetY'$, such that $\pth{\SetX\setminus\SetY}\cup\SetY'$ is a
    hitting set of $\SetBalls$.
\end{compactenum}

\subsubsection{Algorithm}

The input is a set $\SetBalls$ of balls in $\Re^d$ satisfying the
properties above. The algorithm works as follows:
\smallskip%
\begin{compactenum}[\quad(A)]
    \item Compute an initial hitting set $\SetX_0$ of $\SetBalls$ (see
    \pitemref{p:1}), and set $\SetX \leftarrow \SetX_0$.
    
    \item While there is a beneficial exchange of size $\leq \lSize =
    O\pth{1/\eps^d}$ in $\SetX$, carry it out using \pitemref{p:2}.
    
    Specifically, verify for every subset $\SetY \subseteq \SetX$ of
    size at most $\lSize$, if it can be replaced by a strictly smaller
    set $\SetY'$ such that $\pth{\SetX\setminus\SetY}\cup\SetY'$
    remains a hitting set of $\SetBalls$. If so, set $\SetX\leftarrow
    \pth{\SetX\setminus\SetY} \cup\SetY'$, and repeat.

    \item If no such set exists, then return $\SetX$ as the hitting
    set.
\end{compactenum}


\paragraph{Details.}

For the geometric hitting set problem where $\SetBalls$ is a set of
$n$ balls in $\Re^d$, \pitemref{p:1} follows by a simple greedy
algorithm hitting the smallest ball (in the spirit of
\thmrefpage{const:sep}) -- see also \cite{c-ptasp-03}. As for
\pitemref{p:2}, one can check for a smaller hitting set of size at
most $\lSize$ by computing the arrangement $\ArrX{\SetBalls}$, and
directly enumerating all possible hitting sets of size at most
$\lSize$.

The more interesting case is when the set $\SetBalls$ is defined
implicitly by an instance of the minimal Voronoi separation problem
(i.e., we have a set $\PntSet$ of $n$ points in $\Re^d$, and a desired
Voronoi partition $\pth{\PntSet_1,\PntSet_2}$ of $\PntSet$).  Then
\pitemref{p:1} follows from \thmref{const:sep}. Furthermore,
\pitemref{p:2} follows from the algorithm of \thmrefpage{exact} through
computing the minimal hitting set $\SetY'$ of
$\BadBalls{\PntSet\cup\pth[]{\SetX\setminus\SetY}}
{\PntSet_1}{\PntSet_2}$, in $n^{O\pth[]{\lSize}}$ time, where
$\cardin{\SetY'} \leq \lSize$.


\begin{figure}[t]
\AlgorithmBox{%
   \algClusterLocal{}($\SetBalls$, $\Lopt$, $\Opt$):\+\+\\
   \CodeComment{// $\SetBalls$: given set of balls need hitting.}\\
   \CodeComment{// $\Lopt \subseteq \Re^d$: locally optimal hitting
      set of $\SetBalls$ computed by \Algorithm{LocalHitBalls}.}\\
   \CodeComment{// $\Opt \subseteq \Re^d$: optimal hitting set.}%
   \- \\%
   $k \Assign O\pth{1/\eps^d}$, \quad %
   $\SetL_1 \Assign \Lopt$, \quad%
   $\SetO_1 \Assign \Opt$, \quad%
   $\SepSet_1 \Assign \emptyset$, \quad%
   $\SetBalls_1 \Assign \SetBalls$, \quad%
   and $i \Assign 1$.\\
   \While $\SetL_i \cup \SepSet_i \neq \emptyset$ \Do \+\+\\
   Apply \thmrefpage{ball:separator} to $\SetL_i \cup \SepSet_i$ with
   the parameter $\min\pth{ k, \cardin{ \SetL_i \cup \SepSet_i}}$. %
   \\%
   $\ballA_i$, $\oSepSetB_i$: ball and the separator set returned,
   respectively. %
   \\%
   $\leftover_i \Assign \ballA_i\setminus {\bigcup_{j=1}^{i-1}
      \ballA_j}$ be the region of space newly covered by $\ballA_i$.
   \\%
   $\Lopt_i \Assign \Lopt \cap \leftover_i$ and $\Opt_i = \Opt \cap
   \leftover_i$.%
   \\%
   $\SepSet_{i+1} \Assign \pth[]{\SepSet_i \setminus \ballA_i} \cup
   \oSepSetB_i$.%
   \\%
   $\SetL_{i+1} \Assign \SetL_i \setminus \Lopt_i$.%
   \\%
   $i\Assign i+1$%
   \-\-\\
   $\Iterations \Assign i$.%
}%
     \caption{Clustering the local optimal solution.}
     \figlab{cluster:local}
\end{figure}

\subsection{Quality of approximation}

The bound on the quality of approximation follows by a clustering
argument similar to the \PTAS algorithm. The clustering is described
in \figref{cluster:local}. This breaks up the local optimal solution
$\Lopt$ and the optimal solution $\Opt$ into small sets, and we will
use a local exchange argument to bound their corresponding sizes.  The
following lemma testifies that this clustering algorithm provides a
``smooth'' way to move from $\Lopt$ to $\Opt$, through relatively
small steps.

\begin{lemma}
    \lemlab{stabby:stab}%
    Any ball of $\SetBalls$ is stabbed by $\SetX_{i+1} = \SetL_{i+1}
    \cup \SepSet_{i+1} \cup {\bigcup_{j=1}^i \Opt_j}$, for all $i$,
    where $\SetL_{i+1} = {\bigcup_{j=i+1}^\Iterations \Lopt_i}$.
\end{lemma}
\begin{proof}
    The claim holds clearly for $i=0$, as $\SetL_1 = \Lopt$, where
    $\Lopt$ is the locally optimal solution. Now, for the sake of
    contradiction, consider the minimum $i$ for which the claim fails,
    and let $\ballA \in \SetBalls$ be the ball that is not being
    stabbed. We have that $\ballA$ is stabbed by $\SetX_i =
    \pth[]{\bigcup_{j=i}^\Iterations \Lopt_i} \cup \SepSet_{i} \cup
    \pth[]{\bigcup_{j=1}^{i-1} \Opt_j}$ but not by $\SetX_{i+1}$.

    It can not be that $\ballA$ is stabbed by a point of
    $\bigcup_{j=1}^{i-1} \Opt_j$, as such a point is also present in
    $\SetX_{i+1}$.  As such, $\ballA$ must be stabbed by one of the
    points of $\SetL_{i} \cup \SepSet_{i}$, and then this stabbing
    point must be inside $\ballA_i$ -- indeed, the points removed from
    $\SetL_i$ and $\SepSet_{i}$ as we compute $\SetL_{i+1}$ and
    $\SepSet_{i+1}$, respectively, are the ones inside $\ballA_i$.

    Now, if $\ballA$ intersects $\partial \ballA_i$, then $\oSepSetB_i
    \subseteq \SepSet_{i+1}$ stabs $\ballA$ by
    \thmrefpage{ball:separator}, a contradiction.

    So it must be that $\ballA \subseteq \ballA_i$. But then consider
    the region $\Leftover_i = \cup_{j=1}^i \leftover_j = \cup_{j=1}^i
    \ballA_j$. It must be that $\ballA \subseteq \Leftover_i$. This in
    turn implies that the point of $\Opt$ stabbing $\ballA$ is in
    $\Opt \cap \Leftover_i = \cup_{j=1}^i \Opt_i \subseteq
    \SetX_{i+1}$. A contradiction.
\end{proof}

\begin{lemma}
    \lemlab{opt:chunk:small}%
    Consider any ball that $\ballA \in \SetBalls$. Let $i$ be the
    minimum index such that $\Opt_i$ stabs $\ballA$. Then $\ballA$ is
    stabbed by $\Lopt_i \cup \oSepSetB_i \cup \pth[]{ \SepSet_i \cap
       \ballA_i}$. That is $o_i \leq \lChuckSize_i + t_i + z_i$, where
    $o_i = \cardin{\Opt_i}$, $t_i = \cardin{\oSepSetB_i}$, and $z_i =
    \cardin{\SepSet_i \cap \ballA_i}$. Additionally, $o_i =O(k)$, for
    all $i$.
\end{lemma}
\begin{proof}
    If $\ballA$ intersects $\partial \ballA_i$, then it is stabbed by
    $\oSepSetB_i$ by \thmrefpage{ball:separator}. Otherwise, $\ballA
    \subseteq \Leftover_i = \cup_{j=1}^i \leftover_j = \cup_{j=1}^i
    \ballA_j$. In particular, this implies that no later point of
    $\Opt_{i+1} \cup \ldots \cup \Opt_t$ can stab $\ballA$. That is,
    only $\Opt_i$ stabs $\ballA$.  By \lemref{stabby:stab} both
    $\SetX_i$ and $\SetX_{i+1}$ stab $\ballA$, where $\SetX_{i} =
    \pth[]{\Lopt \setminus \Leftover_{i-1}} \cup \SepSet_{i} \cup
    \bigcup_{j=1}^{i-1} \Opt_j$. Namely, $\ballA$ is stabbed by a
    point of $\pth[]{\Lopt \setminus \Leftover_{i-1}} \cup
    \SepSet_{i}$ that is contained inside $\ballA_i$. Such a point is
    either in $\SepSet_{i} \cap \ballA_i$, or in $\pth[]{\Lopt
       \setminus \Leftover_{i-1}} \cap \ballA_i = \Lopt_i$, as
    claimed.

    The second part follows by observing that otherwise $\Opt_i$ can
    be replaced by $\Lopt_i \cup \oSepSetB_i \cup \pth[]{ \SepSet_i
       \cap \ballA_i}$ in the optimal solution.

    The third claim follows by observing that by the algorithm design
    $ \lChuckSize_i + z_i = O(k)$, and $t_i = O\pth[]{ k^{1-1/d}}$. As such,
    $o_i \leq \lChuckSize_i + t_i + z_i=O(k)$.
\end{proof}

\begin{lemma}
    \lemlab{local:is:small}%
    Consider any ball $\ballA \in \SetBalls$ that is stabbed by
    $\Lopt_i$ but it is not stabbed by $\Lopt \setminus \Lopt_i$, then
    $\ballA$ is stabbed by $\Opt_i \cup \oSepSetB_i \cup \pth[]{
       \SepSet_i \cap \ballA_i}$. Additionally, using the notations of
    \lemref{opt:chunk:small}, we have $\lChuckSize_i \leq o_i + t_i + z_i$.
\end{lemma}
\begin{proof}
    If $\ballA$ intersects $\partial \ballA_i$, then it is stabbed by
    $\oSepSetB_i$ by \thmrefpage{ball:separator}.  So we assume for
    now on that $\ballA$ is not stabbed by $\oSepSetB_i$.  But then,
    $\ballA \subseteq \ballA_i \subseteq \Leftover_i = \cup_{j=1}^i
    \leftover_j = \cup_{j=1}^i \ballA_j$.

    Now, by \lemref{stabby:stab} both $\SetX_i$ and $\SetX_{i+1}$ stab
    $\ballA$, where $\SetX_{i+1} = \bigcup_{j=i+1}^\Iterations \Lopt_j
    \cup \SepSet_{i+1} \cup \bigcup_{j=1}^{i} \Opt_j$. Namely,
    $\ballA$ is stabbed by a point of $\SepSet_{i+1} \cup
    \bigcup_{j=1}^{i} \Opt_j$ that is contained inside
    $\ballA_i$. Such a point is either in $\SepSet_{i} \cap \ballA_i$
    and then we are done, or alternatively, it can be in
    $\ballA_i\cap \bigcup_{j=1}^{i} \Opt_j$.

    Now, if $\ballA \subseteq \leftover_i$ then $\Opt_i$ must stab
    $\ballA$, and we are done. Otherwise, let $k < i$ be the maximum
    index such that $\ballA_i$ intersects $\partial \ballA_k$. Observe
    that as $\ballA$ intersects $\leftover_i$, it can not be that it
    intersects the balls $\ballA_{k+1}, \ldots, \ballA_{i-1}$. In
    particular, \thmrefpage{ball:separator} implies that there is a
    point of $\oSepSetB_k$ that stabs $\ballA$, as $\ballA$ is being
    stabbed by $\SetL_k \cup \SepSet_k$. This point of $\oSepSetB_k$
    is added to $\SepSet_{k+1}$, and it is not being removed till
    $\SepSet_{i+1}$. Namely, this point is in $\SepSet_i$, and it is
    of course also in $\ballA_i$, thus implying the claim.

    As for the second part, by \lemref{opt:chunk:small}, $o_i + t_i +
    z_i \leq \alpha = O(k)$. As such, setting $\lSize = \Omega(
    1/\eps^d)$ to be sufficiently large, we have that $\lSize > 2
    \alpha$, and the local search algorithm would consider the local
    exchange of $\Lopt_i$ with $\Opt_i \cup \oSepSetB_i \cup \pth[]{
    \SepSet_i \cap \ballA_i}$. As this is an exchange not taken, it
    must be that it is not beneficial, implying the inequality.
\end{proof}

Now we are left with some easy calculations. 

\begin{lemma}
    \lemlab{silly:calcs}%
    We have  that
    \begin{inparaenum}[(i)]
        \item $\sum_{i=1}^{\Iterations} t_i \leq (\eps/4)
        \cardin{\Opt}$, and 
        \item $\sum_{i=1}^{\Iterations} z_i \leq (\eps/4)
        \cardin{\Opt}$.
    \end{inparaenum}
\end{lemma}

\begin{proof}
    Observe that $k=O(1/\eps^d)$ and $t_i = \cardin{\oSepSetB_i} =
    O(k^{1-1/d}) = O\pth{1/\eps^{d-1}} \leq c \eps k$, where $c$ can be
    made arbitrarily small by making $k$ sufficient large. In
    particular, in every iteration, the algorithm removes $\geq k$
    points from $\SetL_i \cup \SepSet_i$, and replaces them by $t_i$
    points. Starting with a solution of size $\cardin{\Lopt} \leq 
    \cardin{\SetX_0} \leq c'\cardin{\Opt}$, where $c'$ is some constant, 
    this can happen at
    most $\Iterations \leq c' \cardin{\Opt}/ (k - c\eps k) = O( \eps^d
    \cardin{\Opt})$. As such, we have $\sum_{i=1}^{\Iterations} t_i =
    O\pth{ \eps^d \cardin{\Opt} c \eps k } = O\pth{ c\eps
       \cardin{\Opt}}$, as $k = O(1/\eps^d)$. The claim now follows by
    setting $k$ to be sufficiently large.

    The second claim follows by observing that $t_i$ counts the
    numbers of points added to $\SepSet_{i+1}$, while $z_i$ counts the
    number of points removed from it. As $\SepSet_\Iterations$ is
    empty, it must be that $\sum_i t_i = \sum_i z_i$.
\end{proof}

\begin{lemma}
    \lemlab{local:quality}%
    We have that $\cardin{\Lopt} \leq (1+\eps) \cardin{\Opt}$.
\end{lemma}

\begin{proof}
    We have by \lemref{local:is:small} and \lemref{silly:calcs} that
    \begin{align*}
        \cardin{\Lopt}%
        =%
        \sum_{i=1}^\Iterations \lChuckSize_i%
        \leq%
        \sum_{i=1}^\Iterations \pth{o_i + t_i + z_i}%
        =%
        \cardin{\Opt} + \sum_{i=1}^\Iterations t_i +
        \sum_{i=1}^\Iterations z_i%
        \leq%
        (1+\eps/2)\cardin{\Opt}.
    \end{align*}
\end{proof}

\subsubsection{The Result}

\begin{theorem}
    Let $\SetBalls$ be a set of balls in $\Re^d$ satisfying properties
    \pitemrefpage{p:1} and \pitemrefpage{p:2}. Then
    \Algorithm{LocalHitBalls} computes, in $n^{O(1/\eps^d)}$ time, a
    hitting set of $\SetBalls$, whose size is a
    $\pth[]{1+\eps}$-approximation to the minimum size hitting set of
    $\SetBalls$.
\end{theorem}

\begin{proof}
    \lemref{local:quality} implies the bound on the quality of
    approximation. 

    As for the runtime, observe that the size of local solution
    reduces by at least one after each local improvement, and from
    \pitemrefpage{p:1}, the initial local solution has size
    $n^{O\pth[]{1}}$. Thus there can be at most $n^{O\pth[]{1}}$ local
    improvements. Furthermore, every local solution has at most
    $n^{O\pth[]{k}}$ subsets of size $\lSize$ that are checked for local
    improvement. By \pitemrefpage{p:2}, such a local improvement can
    be checked in $n^{O\pth[]{\lSize}}$ time. Thus
    \Algorithm{LocalHitBalls} runs in $n^{O\pth[]{1/\eps^d}}$ time.
\end{proof}

\section{Conclusions}
\seclab{conclusions}

In this paper, we introduced the problem of Voronoi separating a set
of sites into two balanced portions, by introducing separating guards
into the Voronoi diagram. We provided a simple, fast and elegant
algorithm for computing it. We then addressed the problem of how to
compute a Voronoi separating set given a specific partition. This
boils down to a geometric hitting/piercing set problem, which is
quite challenging as the set of balls that needs piercing is
infinite. We showed how to solve this problem exactly using algebraic
techniques, and then we showed how to $(1+\eps)$-approximate it, by
providing a new algorithm for geometric hitting set problem, that
works by reducing the problem into small instances, that can be solved
exactly. We believe this new algorithm to be quite interesting, and it
should be usable for other problems.  Furthermore, the new \PTAS can
be viewed as extending the type of geometric hitting set problems that
can be approximated efficiently. We also introduced a new technique for 
proving local search methods on general geometric hitting set problems, 
based on our strong separator result.

There are many interesting open problems for further research. In 
particular, the new \PTAS might be more practical for 
the piercing balls problem than previous algorithms, and it might be
worthwhile to further investigate this direction. Additionally, the 
proof technique for the local search algorithm might be applicable 
to other separator based geometric problems whether finite or infinite 
in nature.

\newcommand{\etalchar}[1]{$^{#1}$}
 \providecommand{\CNFX}[1]{ {\em{\textrm{(#1)}}}}%
  \providecommand{\CNFSoCG}{\CNFX{SoCG}}%
  \providecommand{\CNFCCCG}{\CNFX{CCCG}}%
  \providecommand{\CNFFOCS}{\CNFX{FOCS}}%
  \providecommand{\CNFSTOC}{\CNFX{STOC}}%
  \providecommand{\CNFSODA}{\CNFX{SODA}}%
  \providecommand{\CNFPODS}{\CNFX{PODS}}%
  \providecommand{\CNFISAAC}{\CNFX{ISAAC}}%
  \providecommand{\CNFFSTTCS}{\CNFX{FSTTCS}}%
  \providecommand{\CNFIJCAI}{\CNFX{IJCAI}}%
  \providecommand{\CNFBROADNETS}{\CNFX{BROADNETS}}%
  \providecommand{\CNFSWAT}{\CNFX{SWAT}} \providecommand{\CNFESA}{\CNFX{ESA}}
  \providecommand{\CNFWADS}{\CNFX{WADS}}%
  \providecommand{\CNFINFOCOM}{\CNFX{INFOCOM}}%
  \providecommand{\Merigot}{M{\'{}e}rigot}%
  \providecommand{\Badoiu}{B\u{a}doiu} \providecommand{\Matousek}{Matou{\v
  s}ek} \providecommand{\Erdos}{Erd{\H o}s}
  \providecommand{\Barany}{B{\'a}r{\'a}ny}
  \providecommand{\Bronimman}{Br{\"o}nnimann}
  \providecommand{\Gartner}{G{\"a}rtner} \providecommand{\Badoiu}{B\u{a}doiu}
  \providecommand{\tildegen}{{\protect\raisebox{-0.1cm}
  {\symbol{'176}\hspace{-0.03cm}}}} \providecommand{\SarielWWWPapersAddr}
  {http://www.uiuc.edu/~sariel/papers} \providecommand{\SarielWWWPapers}
  {http://www.uiuc.edu/\tildegen{}sariel/\hspace{0pt}papers}
  \providecommand{\urlSarielPaper}[1]{ \href{\SarielWWWPapersAddr/#1}
  {\SarielWWWPapers{}/#1}}  \providecommand{\CNFX}[1]{ {\em{\textrm{(#1)}}}}%
  \providecommand{\CNFSTOC}{\CNFX{STOC}}%
  \providecommand{\CNFFOCS}{\CNFX{FOCS}}%


\begin{thebibliography}{{Har}11b}

\bibitem[AGK{\etalchar{+}}01]{agkmp-lshkm-01}
V.~Arya, N.~Garg, R.~Khandekar, K.~Munagala, and V.~Pandit.
\newblock Local search heuristic for $k$-median and facility location problems.
\newblock In {\em Proc. 33rd Annu. ACM Sympos. Theory Comput.\CNFSTOC}, pages
  21--29, 2001.

\bibitem[AKL13]{akl-vddt-13}
F.~Aurenhammer, R.~Klein, and D.-T. Lee.
\newblock {\em Voronoi Diagrams and Delaunay Triangulations}.
\newblock World Scientific, 2013.

\bibitem[AST90]{ast-stgem-90}
\href{http://www.math.tau.ac.il/~nogaa/}{N.~{Alon}}, P.~D. Seymour, and R.~Thomas.
\newblock A separator theorem for graphs with an excluded minor and its
  applications.
\newblock In {\em STOC}, pages 293--299, 1990.

\bibitem[BCKO08]{bcko-cgaa-08}
{M. de} Berg, \href{http://www.win.tue.nl/~ocheong}{O.~{Cheong}}, {M. van} Kreveld, and \href{http://www.cs.uu.nl/people/markov/}{M.~H. {Overmars}}.
\newblock {\em Computational Geometry: Algorithms and Applications}.
\newblock Springer-Verlag, 3rd edition, 2008.

\bibitem[BPR06]{bpr-arag-06}
S.~Basu, R.~Pollack, and M.~F. Roy.
\newblock {\em Algorithms in Real Algebraic Geometry}.
\newblock Algorithms and Computation in Mathematics. Springer, 2006.

\bibitem[BPTW10]{bptw-betsu-10}
J.~{B{\"o}ttcher}, K.~P. {Pruessmann}, A.~{Taraz}, and A.~{W{\"u}rfl}.
\newblock Bandwidth, expansion, treewidth, separators and universality for
  bounded-degree graphs.
\newblock {\em Eur. J. Comb.}, 31(5):1217--1227, July 2010.

\bibitem[CH12]{ch-aamis-12}
\href{http://www.math.uwaterloo.ca/~tmchan/}{T.~M.~{Chan}} and \href{http://www.uiuc.edu/~sariel}{S.~{{Har-Peled}}}.
\newblock Approximation algorithms for maximum independent set of pseudo-disks.
\newblock {\em \href{http://link.springer-ny.com/link/service/journals/00454/}{Discrete Comput. {}Geom.}}, 48:373--392, 2012.

\bibitem[Cha03]{c-ptasp-03}
\href{http://www.math.uwaterloo.ca/~tmchan/}{T.~M.~{Chan}}.
\newblock Polynomial-time approximation schemes for packing and piercing fat
  objects.
\newblock {\em J. Algorithms}, 46(2):178--189, 2003.

\bibitem[EH06]{eh-ggt-06}
\href{http://www.cs.arizona.edu/~alon/}{A.~{Efrat}} and \href{http://www.uiuc.edu/~sariel}{S.~{{Har-Peled}}}.
\newblock Guarding galleries and terrains.
\newblock {\em Inform. Process. Lett.}, 100(6):238--245, 2006.

\bibitem[EJS05]{ejs-ptasg-05}
T.~Erlebach, K.~Jansen, and E.~Seidel.
\newblock Polynomial-time approximation schemes for geometric intersection
  graphs.
\newblock {\em SIAM J. Comput.}, 34(6):1302--1323, 2005.

\bibitem[EK13]{ek-ltamf-13}
D.~Eisenstat and P.~N. Klein.
\newblock Linear-time algorithms for max flow and multiple-source shortest
  paths in unit-weight planar graphs.
\newblock In {\em Proc. 45th Annu. ACM Sympos. Theory Comput.\CNFSTOC}, pages
  735--744, 2013.

\bibitem[FR06]{fr-pgnwe-06}
J.~Fakcharoenphol and S.~Rao.
\newblock Planar graphs, negative weight edges, shortest paths, and near linear
  time.
\newblock {\em J. Comput. Sys. Sci.}, 72(5):868--889, 2006.

\bibitem[GHT84]{ght-stgbg-84}
J.~R. Gilbert, J.~P. Hutchinson, and R.~E. Tarjan.
\newblock A separator theorem for graphs of bounded genus.
\newblock {\em J. Algorithms}, 5(3):391--407, 1984.

\bibitem[GT08]{gt-salsa-08}
A.~Gupta and K.~Tangwongsan.
\newblock Simpler analyses of local search algorithms for facility location.
\newblock {\em CoRR}, abs/0809.2554, 2008.

\bibitem[{Har}11a]{h-gaa-11}
\href{http://www.uiuc.edu/~sariel}{S.~{{Har-Peled}}}.
\newblock {\em Geometric Approximation Algorithms}, volume 173 of {\em
  Mathematical Surveys and Monographs}.
\newblock Amer. Math. Soc., 2011.

\bibitem[{Har}11b]{h-speps-11}
\href{http://www.uiuc.edu/~sariel}{S.~{{Har-Peled}}}.
\newblock A simple proof of the existence of a planar separator.
\newblock {\em CoRR}, abs/1105.0103, 2011.

\bibitem[HL12]{hl-wgscp-12}
\href{http://www.uiuc.edu/~sariel}{S.~{{Har-Peled}}} and M.~Lee.
\newblock Weighted geometric set cover problems revisited.
\newblock {\em J. Comput. Geom.}, 3(1):65--85, 2012.

\bibitem[HM05]{hm-facsk-05}
\href{http://www.uiuc.edu/~sariel}{S.~{{Har-Peled}}} and S.~{Mazumdar}.
\newblock Fast algorithms for computing the smallest $k$-enclosing disc.
\newblock {\em Algorithmica}, 41(3):147--157, 2005.

\bibitem[Kle08]{k-ltast-08}
P.~N. Klein.
\newblock A linear-time approximation scheme for tsp in undirected planar
  graphs with edge-weights.
\newblock {\em SIAM J. Comput.}, 37(6):1926--1952, 2008.

\bibitem[Koe36]{k-kdka-36}
P.~Koebe.
\newblock Kontaktprobleme der konformen {Abbildung}.
\newblock {\em Ber. Verh. S{\"a}chs. Akademie der Wissenschaften Leipzig,
  Math.-Phys. Klasse}, 88:141--164, 1936.

\bibitem[LT77]{lt-stpg-77}
R.~J. Lipton and R.~E. Tarjan.
\newblock A separator theorem for planar graphs.
\newblock In {\em Proc. GI Conf. Theoret. Comput. Sci.}, 1977.

\bibitem[LT79]{lt-stpg-79}
R.~J. Lipton and R.~E. Tarjan.
\newblock A separator theorem for planar graphs.
\newblock {\em SIAM J. Appl. Math.}, 36:177--189, 1979.

\bibitem[MR09]{mr-pghsp-09}
N.~H. Mustafa and S.~Ray.
\newblock {PTAS} for geometric hitting set problems via local search.
\newblock In {\em Proc. 25th Annu. ACM Sympos. Comput. Geom.\CNFSoCG}, pages
  17--22, 2009.

\bibitem[MTTV97]{mttv-sspnng-97}
G.~L. Miller, S.~H. Teng, W.~P. Thurston, and S.~A. Vavasis.
\newblock Separators for sphere-packings and nearest neighbor graphs.
\newblock {\em \href{http://www.acm.org/jacm/}{J. Assoc. Comput. {Mach.}}}, 44(1):1--29, 1997.

\bibitem[O'R87]{o-agta-87}
\href{http://cs.smith.edu/~orourke/}{J.~{O'Rourke}}.
\newblock {\em Art Gallery Theorems and Algorithms}.
\newblock The International Series of Monographs on Computer Science. Oxford
  University Press, 1987.

\bibitem[PA95]{pa-cg-95}
\href{http://www.math.nyu.edu/~pach}{J.~{Pach}} and \href{http://www.cs.duke.edu/~pankaj}{P.~K.~{Agarwal}}.
\newblock {\em Combinatorial Geometry}.
\newblock John Wiley \& Sons, 1995.

\bibitem[SVY09]{svy-dosg-09}
C.~Sommer, E.~Verbin, and W.~Yu.
\newblock Distance oracles for sparse graphs.
\newblock In {\em Proc. 50th Annu. IEEE Sympos. Found. Comput. Sci.\CNFFOCS},
  pages 703--712, Georgia, USA, 2009.

\bibitem[SW98]{sw-gsta-98}
W.~D. Smith and N.~C. Wormald.
\newblock Geometric separator theorems and applications.
\newblock In {\em Proc. 39th Annu. IEEE Sympos. Found. Comput. Sci.\CNFFOCS},
  pages 232--243, 1998.

\end{thebibliography}

\appendix

\section{Separating the guards}
\apndlab{sperating:guards}

Here, we show that given a set of guards $\SetX$ how one can
(efficiently) separates $\SetX$ into two subsets, one of them of size
$\Theta(k)$, by adding $\Theta(k^{1-1/d})$ guards, such that the
underlying ball hitting instance gets separated into two separate
subproblems. To this end, we prove an extension of
\thmrefpage{separator:main}.

\subsection{Preliminaries}

\begin{lemma}
    \lemlab{shield:ball}%
    Given a sphere $\sphereB$ in $\Re^d$ with center $c$ and radius
    $r$, and given a parameter $\lChuckSize, 0 < \lChuckSize \leq r$, one can
    compute in a set $\SetX$ of size $O\pth{\pth{r/\lChuckSize}^{d-1}}$, such
    that every ball of radius $\geq \lChuckSize$ that intersects $\sphereB$
    is stabbed by a point of $\SetX$. The time to compute $\SetX$ is
    $O\pth{\pth{r/\lChuckSize}^{d-1}}$ time.
\end{lemma}

\begin{proof}
    Consider the ring $\Ring = \ballX{c}{r+2\lChuckSize} \setminus
    \ballX{c}{r-2\lChuckSize}$, and pick a $\lChuckSize$-dense set in $\Ring$. To
    this end, consider the grid of sidelength $\lChuckSize/\sqrt{d}$, and
    take the vertices of the grid that are inside $\Ring$. To verify
    the claim on the size, observe that every such grid vertex, must
    be adjacent to a grid cell that is fully contained in the expanded
    ring $\ballX{c}{r+3\lChuckSize} \setminus \ballX{c}{r-3\lChuckSize}$. The volume
    of such a grid cell is $\Omega(\lChuckSize^d)$, and the volume of this
    ring is $O\pth{ r^{d-1} \lChuckSize}$, implying the claim. As for the
    running time claim, observe that the set of vertices of interest
    are grid connected, and a simple \si{BFS} would compute all of
    them.
    
    It is now easy to verify that such a dense set intersects any ball
    of radius $\geq \lChuckSize$ that intersects $\sphereB$.
\end{proof}

Let $\SetBalls$ be a set of balls in $\Re^d$ that are stabbed by a set
$\SetX$ of points in $\Re^d$. Let $k>0$ be a parameter. We show how to
compute a set $\SepSet$ of $O\pth{k^{1-1/d}}$ points and a sphere
$\sphereB$ such that;
\begin{compactenum}[\quad(A)]
    \item The set $\SetX$ is broken into two sets: (i) $\SetX_1$ --
    the points of $\SetX$ inside $\sphereB$, and (ii) $\SetX_2$ -- the
    points outside.
    
    \item $\cardin{\SetX_1} = \Theta(k)$.
    
    \item Every ball of $\SetBalls$ that intersects $\sphereB$, is
    stabbed by a point of $\SepSet$.
\end{compactenum}

\subsection{The Algorithm}

The input is a set of points $\SetX$ in $\Re^d$, and a parameter $k$.
\begin{compactenum}[\qquad (A)]
    \item Using the algorithm of \cite{hm-facsk-05}, compute in linear
    time a $2$-approximation to the smallest (closed) ball
    $\ballX{\cenA}{\radA}$ that contains $k$ points of $\SetX$, where
    $\cenA\in\Re^d$..
    
    \item Let $\radN=\radA / k^{1/d}$.
    
    \item Pick a number $\radB$ uniformly at random from the range
    $\pbrcS{\radA, 2\radA-\radN}$.
    
    \item Let $\SetX'$ represent the points of $\SetX$ in
    $\ballX{\cenA}{2\radA}$.
    
    \item Let $\ballB = \ballX{\cenA}{\radB}$.

    \item Compute the set $\SepSet$ as in \lemref{shield:ball}, of
    size $O\pth{\pth{ \radB / \radN }^{d-1}} = O\pth{k^{1-1/d}}$, for
    the sphere $\sphereB =
    \partial \ballB$ and parameter $l/2$.
    
    \item If a point $\pnt \in \SetX$ is in distance smaller than
    $\radN$ from $\sphereB$, we insert it into $\SepSet$.
\end{compactenum}
\bigskip

\noindent
We claim that in expectation, $\sphereB$ is the desired sphere, and
$\SepSet$ is the corresponding desired set of points.

\subsubsection{Correctness}

\begin{lemma}%
    \lemlab{correctness:1}%
    We have $\cardin{\ballB\cap\SetX} = \Theta(k)$ and
    $\cardin{\SetX'} = \Theta(k)$.
\end{lemma}

\begin{proof}
    Since $\ballB$ contains $\ballX{\cenA}{\radA}$,
    $\cardin{\SetX'}\geq\cardin{\ballB\cap\SetX}\geq k$ by
    construction.  Now by \defref{dbl:constant}, the ball $\ballB$ can
    be covered by $\pth[]{\dblCd}^2$ balls of radius $\radA/2$. Each
    one of them contains at most $k$ points, as $\ballX{\cenA}{\radA}$
    is (up to a factor of two) the smallest ball containing $k$ points
    of $\SetX$. Thus, $\cardin{\ballB\cap\SetX} \leq
    \cardin{\SetX'}\leq \pth[]{\dblCd}^2 k$.
\end{proof}

\begin{lemma}
    \lemlab{sphere:separates}%
    Let $\SetBalls$ be a set of balls stabbed by $\SetX$.  Every ball
    in $\SetBalls$ that intersects $\sphereB$, is stabbed by a point
    of $\SepSet$.
\end{lemma}

\begin{proof}
    By construction, every ball of $\SetBalls$ that intersects
    $\sphereB$, whose size is at least $\radN/2$, is stabbed by a
    point of $\SepSet$.  For the case of smaller balls, consider that
    every ball in $\SetBalls$ is stabbed by some point of $\SetX$. Now
    any ball of size less than $\radN/2$ intersecting $\sphereB$, must
    be centered within a distance of $\radN/2$ from
    $\sphereB$. Consequently, the point of $\SetX$ that stabs it, must
    be within a distance of $\radN$ from $\sphereB$. But by
    construction, $\SepSet$ also contains this point.
\end{proof}

\begin{lemma}
    \lemlab{separator:size}%
    Let $Y = \cardin{\SepSet}$. We have that $\Ex{Y} =
    O\pth{k^{1-1/d}}$.
\end{lemma}

\begin{proof}
    Let $Z$ be the number of points of $\SetX$, that were added to
    $\SepSet$. We claim that $\Ex{Z} = O\pth{k^{1-1/d}}$.  To this
    end, for a point $\pnt \in \SetX'$, let $X_\pnt$ be the indicator
    variable that is one if and only if $\pnt$ is in distance at most
    $\radN$ from $\sphereB$. The algorithm picked the radius $\radB$
    uniformly at random in the interval $[\radA,
    2\radA-\radN]$. Furthermore, $X_\pnt= 1$ if and only if
    $\dist{\pnt}{\cenA} - \radN \leq \radB \leq \dist{\pnt}{\cenA} +
    \radN$. This happens only if $\radB$ falls into an interval
    $I_\pnt$ that is of length at most $2\radN$ centered at
    $\dist{\pnt}{\cenA}$. As such, we have
    \begin{equation*}
        \Prob{\MakeBig X_{\pnt}=1} %
        =%
        \frac{ \lengthX{ \MakeSBig  I_\pnt \cap [r,2r-\radN] }}{\lengthX{
              \MakeSBig [r,2r-\radN]}} %
        \leq%
        \frac{2\radN}{r-\radN}%
        \leq%
        \frac{3\radN}{\radA}%
        = %
        \frac{3\pth{\radA / k^{1/d}}}{\radA} \leq \frac{3}{k^{1/d}}.
    \end{equation*}
    Now, by linearity of expectation, we have that $\Ex{Z} =
    \sum_{\pnt \in \SetX'} \Ex{ X_\pnt } \leq 3\cardin{\SetX'}/
    {k^{1/d}}\leq 3\pth[]{\dblCd}^2 k^{1-1/d}$. This implies the
    claim, as $Y = Z + O\pth{n^{1-1/d} }$.
\end{proof}

\subsection{The result}
\begin{theorem}[Restatement of \thmref{ball:separator}.]
    \thmlab{ball:separator:restate}%
    Let $\SetX$ be a set of points in $\Re^d$, and $k>0$ be an integer
    sufficiently smaller than $\cardin{\SetX}$.  One can compute, in
    $O\pth{\cardin{\SetX}}$ expected time, a set $\SepSet$ of
    $O\pth{k^{1-1/d}}$ points and a sphere $\sphereB$ containing
    $\Theta(k)$ points of $\SetX$ inside it, such that for any set
    $\SetBalls$ of balls stabbed by $\SetX$, we have that every ball
    of $\SetBalls$ that intersects $\sphereB$ is stabbed by a point of
    $\SepSet$.
\end{theorem}


\end{document}